\documentclass[10pt,twocolumn,twoside]{IEEEtran} 
\usepackage{lipsum} % for filler text
\usepackage{url}
\usepackage{mathrsfs}
\usepackage{graphics} % for pdf, bitmapped graphics files
\usepackage{epsfig}   % for postscript graphics files
\usepackage{mathptmx} % assumes new font selection scheme installed
\usepackage{times}    %assumes new font selection scheme installed
\usepackage{amsmath}  % assumes amsmath package installed
\usepackage{amsthm}
\usepackage{amssymb}  % assumes amsmath package installed
\usepackage{graphicx}
\usepackage{epstopdf}
\usepackage{bm}
\usepackage{enumerate}
\usepackage{subfig}
\usepackage[noadjust]{cite}
\usepackage{multirow}
\usepackage[english]{babel}
\usepackage[autostyle]{csquotes}
\usepackage{xcolor}
\usepackage{makecell}
\usepackage{threeparttable}
\usepackage{amsfonts}
\usepackage{mathrsfs}
\usepackage{chngcntr}
\usepackage{setspace}
\usepackage{float} 
\usepackage{caption}
\usepackage{color}
\usepackage{epstopdf}
\usepackage{booktabs}
\usepackage{slashbox}

\newtheorem{Theorem}{Theorem}
\newtheorem{Corollary}{Corollary}
\newtheorem{Remark}{Remark}
\newtheorem{Lemma}{Lemma}

\begin{document}
\title{An Attack-Resilient Pulse-Based Synchronization Strategy for General Connected Topologies}

\author{Zhenqian~Wang and %~\IEEEmembership{Member,~IEEE,}
        Yongqiang~Wang,~\IEEEmembership{Senior Member,~IEEE}
\thanks{Zhenqian Wang and Yongqiang Wang are with the Department of Electrical and Computer Engineering, Clemson University, Clemson, SC, 29630 USA. e-mail:~\{zhenqiw, yongqiw\}@clemson.edu. The work was supported in part by the National Science Foundation under Grants 1738902, 1912702, and in part by the China Scholarship Council.}}% <-this % stops a space

% \author{\cdotslastname \thanks{\cdots} \thanks{\cdots} }

\markboth{}%
{}
\maketitle

\begin{abstract}
 Synchronization of pulse-coupled oscillators (PCOs) has gained significant attention recently due to increased applications in sensor networks and wireless communications. Given the distributed and unattended nature of wireless sensor networks, it is imperative to enhance the resilience of pulse-based synchronization against malicious attacks. However, most existing results on resilient PCO synchronization are obtained for all-to-all networks. We propose a new pulse-based synchronization mechanism to improve the resilience of PCO synchronization that is applicable under general connected topologies. Under the proposed synchronization mechanism, we rigorously characterize the condition for stealthy Byzantine attacks and prove that perfect synchronization of legitimate oscillators can be guaranteed in the presence of multiple stealthy Byzantine attackers, irrespective of whether the attackers collude with each other or not. The new mechanism can guarantee resilient synchronization even when the initial phases of legitimate oscillators are widely distributed in a half circle, which is in distinct difference from most existing attack-resilient synchronization algorithms (including the seminal paper from Lamport and Melliar-Smith \cite{lamport1985synchronizing}) that require a priori (almost) synchronization among legitimate oscillators. Numerical simulation results are given to confirm the theoretical results.
\end{abstract}
\begin{IEEEkeywords}
Synchronization, Pulse-Coupled Oscillators, General Connected Topologies, Stealthy Byzantine Attacks. 
\end{IEEEkeywords}
%\hfill Zhenqian Wang
%\hfill October 06, 2018
%\IEEEpeerreviewmaketitle

\section{Introduction}

Inspired by flashing fireflies and contracting cardiac cells, pulse-based synchronization is attracting increased attention in sensor networks and wireless communications \cite{mirollo1990synchronization,peskin1975mathematical,mathar1996pulse,simeone2008distributed}. By exchanging simple and identical messages (so-called pulses), pulse-based synchronization incurs much less energy consumption and communication overhead compared with conventional packet-based synchronization approaches \cite{pagliari2011scalable}. These inherent advantages make pulse-based synchronization extremely appealing for event coordination and clock synchronization in various networks \cite{werner2005firefly,hong2005scalable,hu2006scalability,leidenfrost2009firefly,nunez2017pulse}. In the past decade, plenty of results have been reported on pulse-based synchronization. For example, by optimizing the interaction function, i.e., phase response function, the synchronization speed of pulse-coupled oscillators (PCOs) is maximized in \cite{wang_tsp2:12}; with a judiciously-added refractory period in the phase response function, the energy consumption of PCO synchronization is reduced in \cite{konishi:08,okuda2011experimental,wang_tsp:12}; \cite{wang_tsp:13,nunez2015synchronization,nunez2016synchronization} show that PCOs can achieve synchronization under a general coupling topology even when their initial phases are randomly distributed in the entire oscillation period. Recently, synchronization of PCOs in the presence of time-delays and unreliable links is also discussed \cite{klinglmayr2012guaranteeing,klinglmayr2017convergence}. Other relevant results include \cite{canavier2010pulse,nishimura2011robust,nishimura2012probabilistic,lucken2012two,nunez2015global,kannapan2016synchronization,lyu2018global,proskurnikov2017synchronization,gao2018pulse}.

However, the above results are obtained under the assumption that all oscillators behave correctly with no nodes compromised by malicious attackers. Due to the distributed and unattended nature, wireless sensor nodes are extremely vulnerable to attacks, making it imperative to study synchronization in the presence of attacks. Although plenty of discussions exist for conventional packet-based synchronization, e.g.,  \cite{pease1980reaching,lamport1985synchronizing,manzo2005time,li2006global,song2007attack,du2008security,leidenfrost2010fault}, results are very sparse on the attack-resilience of pulse-based synchronization. In \cite{tyrrell2010does}, the authors showed that pulse-based synchronization is more robust than its packet-based counterpart in the presence of a faulty node. In \cite{klinglmayr2012self}, a new phase response function was proposed to improve the precision of pulse-based synchronization against non-persistent random attacks. The authors in \cite{yun2015robustness} considered pulse-based synchronization in the presence of faulty nodes which fire periodically ignoring neighboring nodes' influence. However, none of the above results address phase synchronization of PCOs when compromised nodes act maliciously to corrupt synchronization by applying disturbing pulses with judiciously-crafted patterns. Furthermore, the above results only apply to a priori synchronized PCOs, i.e., all legitimate nodes are required to have identical phases when faulty pulses are emitted.

In this paper, we present a new pulse-based synchronization strategy for general connected PCOs that can achieve phase synchronization even in the presence of multiple stealthy Byzantine attackers. Throughout this paper, we use ``general connected" to describe undirected graphs in which there exists a (multi-hop) path between any pair of nodes. In the pulse-based interaction framework where exchanged messages are identical and content-free, Byzantine attacks mean compromised nodes injecting pulses using judiciously crafted patterns to disturb the synchronization process. So compared with existing results in \cite{tyrrell2010does,klinglmayr2012self,yun2015robustness} which address faulty PCO nodes with random or periodic pulse emitting patterns, the situation considered in this paper is more difficult to deal with due to the intelligent behavior of malicious attackers. By proposing a new pulse-based interaction mechanism, we show that perfect phase synchronization of legitimate oscillators can still be guaranteed as long as their initial phases are distributed within a half oscillation period. The approach is applicable even when individual oscillators do not have access to the total number of oscillators in a network. The result is in distinct difference from our recent results in \cite{wang2018pulse,wang2018attack} which can only guarantee phase synchronization under all-to-all topologies. 

The main contributions of this paper are as follows: 1) We propose a new mechanism for pulse-coupled synchronization that employs a ``cut-off" algorithm to restrict the number of pulses able to affect a receiving oscillator's phase in any three-quarter oscillation period, which is key to enable resilience to attacks; 2) The ``cut-off" algorithm also brings superior robustness to time-varying delays (see the numerical-simulation based comparison with existing algorithms in the absence of attacks in Fig. \ref{Sync_error_1} and Fig. \ref{Sync_error_2}), making the new pulse-coupled synchronization mechanism fundamentally different from existing ones and important in its own even in the absence of attacks; 3) We rigorously analyze the condition for an attacker to stay stealthy in a general connected pulse-coupled oscillator network, and address an attack model that is more difficult to deal with than existing results like \cite{wang2018pulse,wang2018attack}; 4) We guarantee that the collective oscillation period is invariant under attacks and identical to the free-running period, which is superior to existing results (e.g., \cite{wang2018pulse,wang2018attack}) that lead to a collective oscillation period affected by attacker pulses; 5) The results are applicable to general connected topologies whereas existing results on attack-resilience of pulse-coupled synchronization all assume an all-to-all topology.

It is worth noting that the analysis method here is also significantly different from the methods in \cite{wang2018pulse,wang2018attack}. In \cite{wang2018pulse,wang2018attack}, one can prove that the length of the containing arc will decrease to a value no greater than $(1-l)$ of its original value after each round of firing, where $l\in(0,\,1]$ is the coupling strength. However, in this paper, while enabling resilience to attacks, the new interaction mechanism also leads to more complicated dynamics, as reflected by the fact that we cannot prove length reduction in the containing arc after each round of firing. In fact, in the worse case, we can only prove that the length of the containing arc will decrease to a value no greater than $(1-l/2)$ of its original value after every two consecutive firing rounds.

This paper is organized as follows. Sec. \uppercase\expandafter{\romannumeral2} introduces a new pulse-based synchronization mechanism. Under the new mechanism, Sec. \uppercase\expandafter{\romannumeral3} presents a synchronization condition for general connected PCOs in the absence of attacks. In Sec. \uppercase\expandafter{\romannumeral4}, we characterize the condition for an attacker to keep stealthy, i.e., mounting attacks without being detected. In Sec. \uppercase\expandafter{\romannumeral5}, we prove that synchronization of legitimate oscillators can be guaranteed in the presence of multiple stealthy Byzantine attackers, with and without collusion. In Sec. \uppercase\expandafter{\romannumeral6}, we prove the applicability of our approach even when the total number of oscillators is unknown to individual oscillators. Simulation results are presented in Sec. \uppercase\expandafter{\romannumeral7}.

\section{A New Pulse-Based Synchronization Mechanism}
Consider a network of $N$ pulse-coupled oscillators. Each oscillator is equipped with a phase variable. When the evolving phase of an oscillator reaches $2\pi$ rad, the oscillator emits a pulse. Receiving pulses from neighboring oscillators will lead to the adjustment of the receiving oscillator's phase, which can be designed to achieve a desired collective behavior such as phase synchronization. An edge $(i,j)$ from oscillator $i$ to oscillator $j$ means that oscillator $j$ can receive pulses from oscillator $i$ but not necessarily vice versa. The number of edges entering oscillator $i$ is called the indegree of oscillator $i$ and is represented as $d^-(i)$. The number of edges leaving oscillator $i$ is called the outdegree of oscillator $i$ and is represented as $d^+(i)$. The value $d(i) \triangleq\min\{d^-(i),d^+(i)\}$ is called the degree of oscillator $i$. The degree of a network is defined as $d\triangleq\min_{i=1,2,\cdots,N}\{d(i)\}$.

The conventional pulse-based synchronization mechanism is presented below:

\noindent\rule{9cm}{0.12em}
\emph{Conventional Pulse-Based Synchronization Mechanism \cite{yun2015robustness}:} \\
\rule{9cm}{0.1em}
\begin{enumerate}
\item The phase $\phi_i$ of oscillator $i$ evolves from $0$ to $2\pi$ rad with a constant speed $\omega=1$ rad/second.

\item Once $\phi_i$ reaches $2\pi$ rad, oscillator $i$ fires and resets its phase to $0$.
		
\item \emph{Whenever} oscillator $i$ receives a pulse, it \emph{instantaneously} resets its phase to:
\begin{equation}\label{PhaseJump}
\phi_i^+=\phi_i+l\times F(\phi_i)\\
\end{equation}
where $l\in(0,1]$ is the coupling strength and $F(\bullet)$ is the phase response function (PRF) given below:
\begin{equation}\label{PRF}
\begin{array}{ccc}
F(\phi):=\left\{
\begin{array}{ll}
\hspace{0.1cm} -\phi~~~~~~~~~0\leq\phi\leq\pi \\
\hspace{0.1cm} 2\pi-\phi~~~~~\pi<\phi\leq 2\pi 
\end{array}\right.
\end{array}
\end{equation}
\end{enumerate}
\rule{9cm}{0.12em}
%%%%%%%%%%%%%%%%%%%%%%%%%%%%%%%%%%%%%%%%%%%%%%%%%%%%%%%%%%%%%%%%%%%%%%%%%%%%%%%%%%%%%

In the above conventional pulse-based synchronization mechanism, every incoming pulse will trigger a jump on the receiving oscillator's phase, which makes it easy for attackers to perturb the phases of legitimate oscillators and destroy their synchronization. Moreover, one can easily get that synchronization can never be maintained for general connected PCOs under the conventional mechanism, even when the coupling strength is set to $l=1$. This is because attack pulses can always exert nonzero phase shifts on affected legitimate oscillators and make them deviate from unaffected ones. Due to the same reason, existing attack resilient pulse-coupled synchronization mechanisms in \cite{wang2018pulse} and \cite{wang2018attack} for all-to-all graphs cannot be applied to general connected graphs, either. Motivated by these observations on the inherent vulnerability of existing pulse-based synchronization mechanisms, we propose a new pulse-based synchronization mechanism to improve the attack resilience of general connected PCO networks. Our key idea to enable attack resilience is a ``cut-off" mechanism which can restrict the number of pulses able to affect a receiving oscillator's phase in any three-quarter oscillation period. The ``cut-off" mechanism only allows pulses meeting certain conditions to affect a receiving oscillator's phase and hence can effectively filter out attack pulses with extremely negative effects on the synchronization process. Noting that all pulses are identical and content-free, so the ``cut-off" mechanism is judiciously designed based on the number of pulses an oscillator received in the past, i.e., based on memory. This is also the reason that we let an entire oscillation period $T=2\pi$ seconds elapse so that each oscillator can acquire memory.
 
\noindent\rule{9cm}{0.12em}
\emph{New Pulse-Based Synchronization Mechanism (Mechanism 1):}\\
\rule{9cm}{0.1em}
\begin{enumerate}
\item The phase $\phi_i$ of oscillator $i$ evolves from $0$ to $2\pi$ rad with a constant speed $\omega=1$ rad/second.
 	
\item Once $\phi_i$ reaches $2\pi$ rad, oscillator $i$ fires and resets its phase to $0$.
 	
\item When oscillator $i$ receives a pulse at time instant $t$, it resets its phase according to (\ref{PhaseJump}) only when all the following three conditions are satisfied:
\begin{enumerate}
\item an entire period of $T=2\pi$ seconds has elapsed since initiation.
 		
\item before receiving the current pulse, oscillator $i$ has received at least
\begin{equation}\label{lower_bound}
\lambda_i=\lfloor(d(i)-\lfloor N/2 \rfloor)/4\rfloor
\end{equation}
pulses within $(t-T/4,\,t]$, where $d(i)$ is the degree of oscillator $i$ and $\lfloor\bullet\rfloor$ is the largest integer no greater than $``\bullet."$
\item  before receiving the current pulse, oscillator $i$ has received less than $\bar{\lambda}_i$ pulses within $(t-3T/4,\,t]$, where
\begin{equation}\label{upper_bound}	
\bar{\lambda}_i=d(i)-2\lambda_i
\end{equation}
\end{enumerate}
Otherwise, the pulse has no effect on $\phi_i$.
\end{enumerate}
\rule{9cm}{0.12em}
 
Fig. \ref{fig_PRM} illustrates the phase evolution of oscillator $i$ having degree $d(i)=9$ in a network of $11$ PCOs. According to (\ref{lower_bound}) and (\ref{upper_bound}), we have $\lambda_i=1$ and $\bar{\lambda}_i=7$. So a pulse received at time instant $t$ can shift oscillator $i$'s phase when all the following three conditions are met: $1)$ $t>T$; $2)$ oscillator $i$ has received at least $1$ pulse within $(t-T/4,\,t]$; and $3)$ oscillator $i$ has received less than $7$ pulses within $(t-3T/4,\,t]$. Take the scenario in Fig. \ref{fig_PRM} as an example, only the $11$th and the $12$th pulses triggered phase jumps on oscillator $i$.
\begin{figure}[htbp]
\centering
\includegraphics[width=0.4\textwidth]{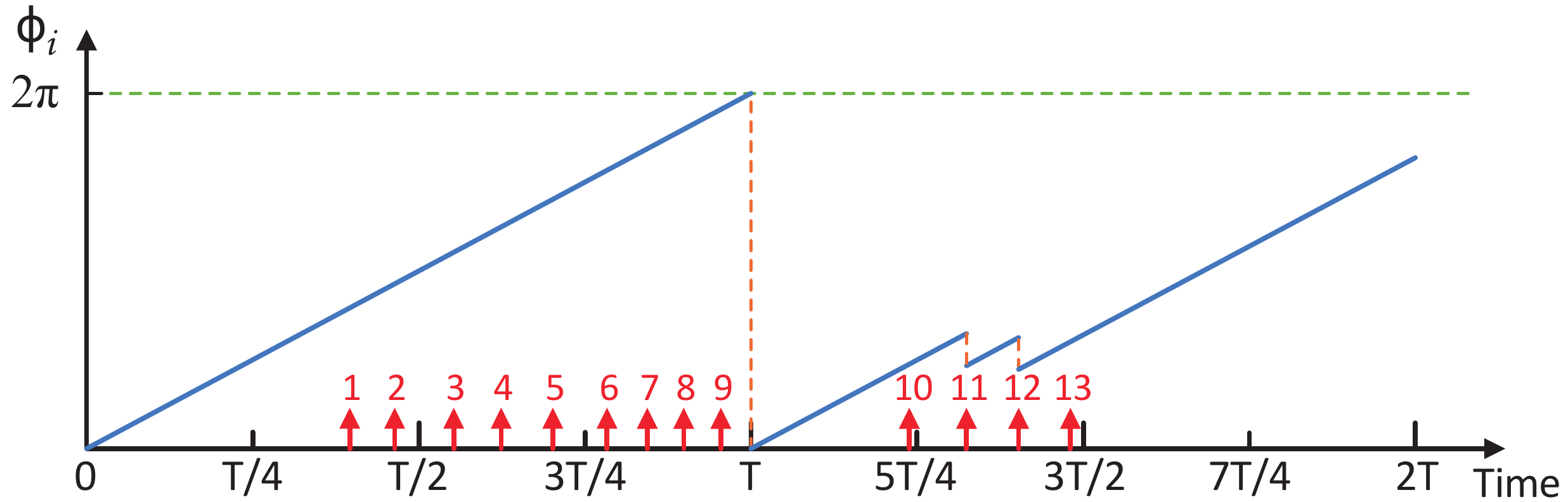}
\caption{The phase evolution of oscillator $i$ in a network of $11$ PCOs under Mechanism $1$. Indexed red arrows represent incoming pulses.}
\label{fig_PRM}
\end{figure}
%%%%%%%%%%%%%%%%%%%%%%%%%%%%%%%%%%%%%%%%%%%%%%%%%%%%%%%%%%%%%%%%%%%%%%%%%%%
\begin{Remark}\label{Remark_1}
Following \cite{lucken2012two,nunez2015global, proskurnikov2017synchronization}, we assume that when a legitimate oscillator receives multiple pulses simultaneously, it will process the incoming pulses consecutively. In other words, no two pulses will be regarded as an aggregated pulse.  
\end{Remark}
%%%%%%%%%%%%%%%%%%%%%%%%%%%%%%%%%%%%%%%%%%%%%%%%%%%%%%%%%%%%%%%%%%%%%%%%%%%

\section{Synchronization of General Connected PCOs in the Absence of Attacks}

In this section, we will show that Mechanism $1$ can guarantee the synchronization of general connected PCOs in the absence of attacks.

Assuming that all oscillators' phases rotate clockwise on a unit circle, the containing arc of legitimate oscillators is defined as the shortest arc on the unit circle that contains all legitimate oscillators' phases. The leading and terminating points of a containing arc are defined as the starting and ending points of the containing arc in the clockwise direction, respectively.

Based on the definition of containing arc, we can define phase synchronization:

\emph{Definition 1 (Phase Synchronization)}: A network of pulse-coupled oscillators achieves phase synchronization if the length of the containing arc of all legitimate oscillators converges to $0$ upon which all legitimate oscillators fire simultaneously with a fixed period $T=2\pi$ seconds.

\begin{Remark}
Requiring the firing period to be $T=2\pi$ seconds in Definition $1$ is important for two reasons. First, this requirement guarantees that all legitimate oscillators will not have irregular behaviors. For example, otherwise all oscillators having fixed and constant phases $0$ meets the condition of containing arc converging to $0$ but is unacceptable for pulse-coupled oscillators. Secondly, this additional requirement on firing period guarantees that the collective oscillation period after synchronization is not affected by attacks. In fact, in existing results \cite{klinglmayr2012self,yun2015robustness,wang2018pulse,wang2018attack}, the collective firing period could be affected by attack pulses.
\end{Remark}

We next give two important properties of general connected PCO networks under Mechanism $1$.

\begin{Lemma}\label{Lemma_1}
For a general connected network of $N$ legitimate PCOs evolving under Mechanism $1$, when the initial length of the containing arc is less than $\pi$ rad, the length of the containing arc is non-increasing. 
\end{Lemma}

\begin{proof} Following the same line of reasoning as in Theorem $1$ of \cite{wang2018pulse}, the containing arc's length will change only when an oscillator's firing triggers a phase jump on at least one other oscillator. We assume that oscillator $i$ fires at time instant $t_i$ whose pulse triggers a phase jump on at least one other oscillator. One can easily get $\phi_i(t_i)=2\pi$ and the phase distribution of all the other $N-1$ oscillators can only fall within one of the following three scenarios, as depicted in Fig. \ref{fig_noattack}:
\begin{itemize}
\item[a)] all the other $N-1$ oscillators' phases reside in $(\pi, 2\pi]$;
\item[b)] all the other $N-1$ oscillators' phases reside in $[0, \pi)$;
\item[c)] the other $N-1$ oscillators' phases reside partially in $[0, \pi)$ and partially in $(\pi, 2\pi]$.
\end{itemize}

\begin{figure}[htbp]
\centering
\includegraphics[width=0.4\textwidth]{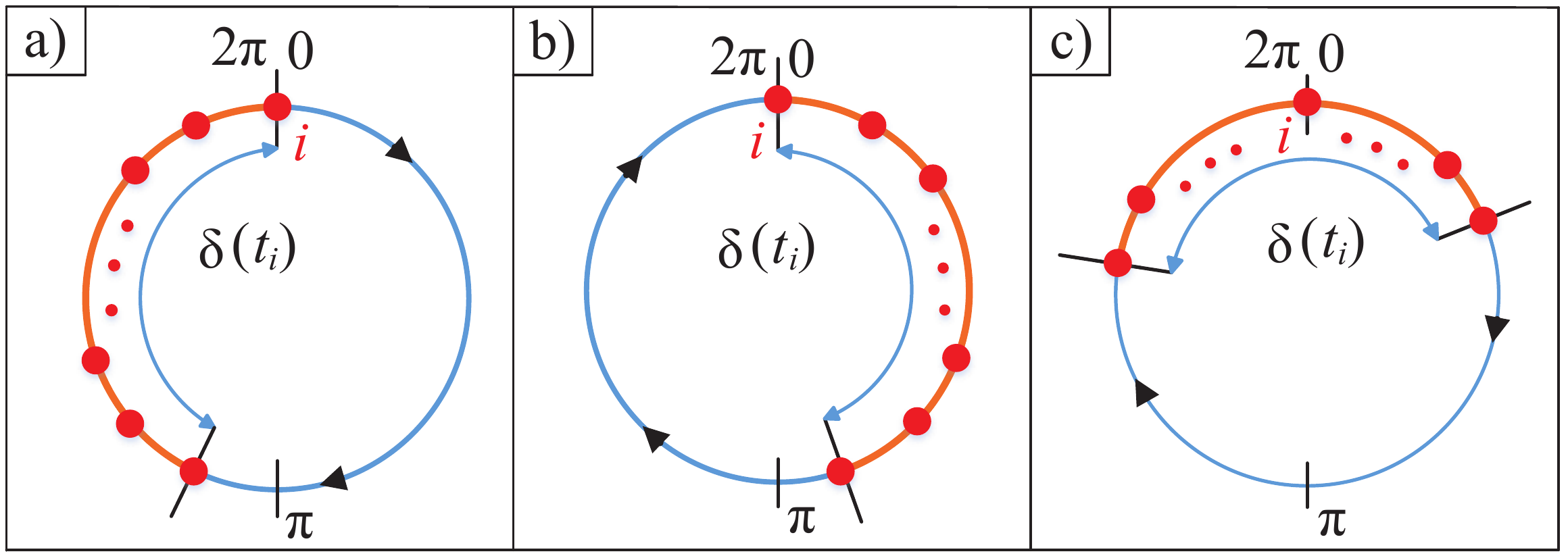}
\caption{Three scenarios of phase distributions of oscillators when oscillator $i$ fires at time instant $t_i$.}	
\label{fig_noattack}
\end{figure}

Denoting $\delta(t_i)$ as the length of the containing arc at time instant $t_i$, next we show that $\delta(t_i)$ cannot be increased by the firing of oscillator $i$ in any of the aforementioned three scenarios, i.e., $\delta^+(t_i)\leq\delta(t_i)$ always holds.
\begin{itemize}			
\item[a)] When all the other $N-1$ oscillators' phases reside in $(\pi, 2\pi]$ at $t_i$, the length of the containing arc can be expressed as
\begin{flalign}\label{Length_1.1}
\delta(t_i)=2\pi-\min_{j\in\mathcal{N},j\neq i}\{\phi_j(t_i)\}
\end{flalign}
where $\mathcal{N}=\{1,2,\cdots,N\}$ is the index set of all oscillators. After the firing of oscillator $i$, we have $\phi_i^+(t_i)=0$. Since the PRF in (\ref{PRF}) is non-negative on $(\pi, 2\pi]$, the pulse can only trigger a forward jump or have no effect on an oscillator with phase residing in $(\pi, 2\pi]$. Hence, we have $\phi_j^+(t_i)=\phi_j(t_i)+F(\phi_j(t_i))\geq\phi_j(t_i)$ or $\phi_j^+(t_i)=\phi_j(t_i)$ for $j\in\mathcal{N}, j\neq i$. In both cases we have $\phi_j(t_i)\leq\phi_j^+(t_i)$ for $j\in\mathcal{N}, j\neq i$, which implies
\begin{flalign}\label{Length_1.2}
\min_{j\in\mathcal{N},j\neq i}\{\phi_j(t_i)\}\leq\min_{j\in\mathcal{N}, j\neq i}\{\phi_j^+(t_i)\}
\end{flalign}
The length of the containing arc immediately after oscillator $i$'s firing at $t_i$ becomes
\begin{flalign}\label{Length_1.3}
\delta^+(t_i)&=2\pi-\min_{j\in\mathcal{N}, j\neq i}\{\phi_j^+(t_i)\}+\phi_i^+(t_i)\nonumber\\
&=2\pi-\min_{j\in\mathcal{N}, j\neq i}\{\phi_j^+(t_i)\}
\end{flalign}
One can easily get $\delta^+(t_i)\leq\delta(t_i)$ by combining (\ref{Length_1.1}), (\ref{Length_1.2}) and (\ref{Length_1.3}).

\item[b)] When all the other $N-1$ oscillators' phases reside in $[0, \pi)$ at time instant $t_i$ (note that phases $0$ and $2\pi$ are the same point on the unit circle), noting that under Mechanism $1$, the pulse can only trigger a backward jump or have no effect on an oscillator with phase residing in $[0,\,\pi)$, one can easily get $\delta^+(t_i)\leq\delta(t_i)$ following the same line of reasoning as in Scenario $a)$.

\item[c)] When the other $N-1$ oscillators' phases reside partially in $(\pi, 2\pi]$ and partially in $[0, \pi)$ at time instant $t_i$, one can easily get $\delta^+(t_i)\leq\delta(t_i)$ by combining the arguments in Scenario $a)$ and Scenario $b)$.
\end{itemize}

Summarizing the above three scenarios, we get that the length of the containing arc is non-increasing.
\end{proof}

Based on Lemma \ref{Lemma_1}, next we show that every oscillator will fire at least once within any time interval of length $3T/2$ under Mechanism $1$.

\begin{Lemma}\label{Lemma_2}
For a general connected network of $N$ legitimate PCOs with their initial length of the containing arc less than $\pi$ rad, every oscillator will fire at least once within any time interval of length $3T/2$ under Mechanism $1$.
\end{Lemma}

\begin{proof} From Lemma $1$, we know that the length of the containing arc is non-increasing. So the phase distribution of all oscillators at an arbitrary time instant $t$ can only fall within one of the following four scenarios, as illustrated in Fig. \ref{keep_firing}:	
\begin{itemize}
	\item[1)] all oscillators' phases reside in $[0,\,\pi]$;
	\item[2)] oscillators' phases reside partially in $[0,\,\pi]$, partially in $(\pi,\,2\pi]$ and the containing arc includes phase $\pi$;
	\item[3)] all oscillators' phases reside in $(\pi,\,2\pi]$;
	\item[4)] oscillators' phases reside partially in $[0,\,\pi]$, partially in $(\pi,\,2\pi]$ and the containing arc includes phase $2\pi$.
\end{itemize}
\begin{figure}[htbp]
	\centering
	\includegraphics[width=0.45\textwidth]{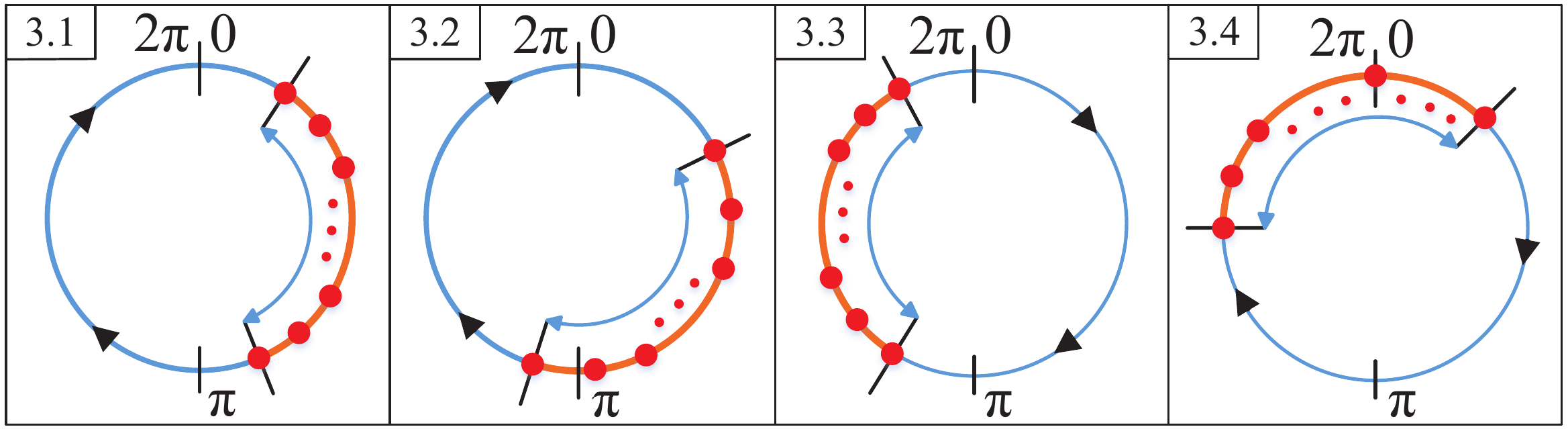}
	\caption{Four possible scenarios of phase distribution at time instant $t$.}	
	\label{keep_firing}
\end{figure}

Since all oscillators are legitimate, according to Mechanism $1$, one can easily get that in Scenarios $1)$, $2)$ and $3)$, all oscillators will evolve towards phase $2\pi$ rad and fire within $[t,\,t+T]$. In Scenario $4)$, given that the PRF in (\ref{PRF}) is non-negative on $(\pi, 2\pi]$, the pulse can only advance or have no effect on the oscillators with phase residing in $(\pi, 2\pi]$. Hence, all oscillators residing in $(\pi, 2\pi]$ will evolve towards phase $2\pi$ rad and fire within $[t,\,t+T/2]$. Since the length of the containing arc is less than $\pi$ rad and non-increasing, all oscillators reside in $[0,\pi]$ immediately after the firing of the oscillator on the ending point of the containing arc, meaning that the network shifts to Scenario $1)$. Then all oscillators will evolve towards phase $2\pi$ rad and fire within the following $T$ seconds. Therefore, we can get that in Scenario $4)$, every oscillator will fire within $[t,\,t+3T/2]$. By iterating the above argument, we know that every oscillator will fire at least once within any time interval of length $3T/2$.
\end{proof}

Now we are in position to present the synchronization condition in the absence of attacks:

\begin{Theorem}\label{Theorem_1_WithoutAttack}
For a general connected network of $N$ legitimate PCOs, if the initial length of the containing arc is less than $\pi$ rad and the degree of the PCO network satisfies $d>\lfloor N/2\rfloor$, then the containing arc of all oscillators will converge to zero under Mechanism $1$.
\end{Theorem}

%%%%%%%%%%%%%%%%%%%%%%%%%%%%%%%%%%%%%%%%%%%%%%%%%%%%%%%%%%%%%%%%%%%%%%%%%%%%%%%%%%%%%%%%%%
\begin{proof} Without loss of generality, we denote $\delta(t)$ as the length of the containing arc at time $t$ and set the initial time to $t=0$. According to Lemma \ref{Lemma_1}, we have that the containing arc is non-increasing and $0\leq\delta(t)<\pi$ for $t\geq 0$. From Lemma \ref{Lemma_2}, every oscillator will fire at least once within any time interval of length $3T/2$ and hence there exists a time instant $t_0>2T$ at which the ending point of the containing arc resides at phase $0$. Denoting the starting point of the containing arc at this time instant as $0\leq\epsilon<\pi$, we have $\delta(t_0)=\epsilon$. Next, we separately discuss the $0\leq \epsilon<\pi/2$ case and the $\pi/2\leq\epsilon<\pi$ case to prove the convergence of $\delta(t)$ to $0$.

\begin{figure}[htbp]
\centering
\includegraphics[width=0.35\textwidth]{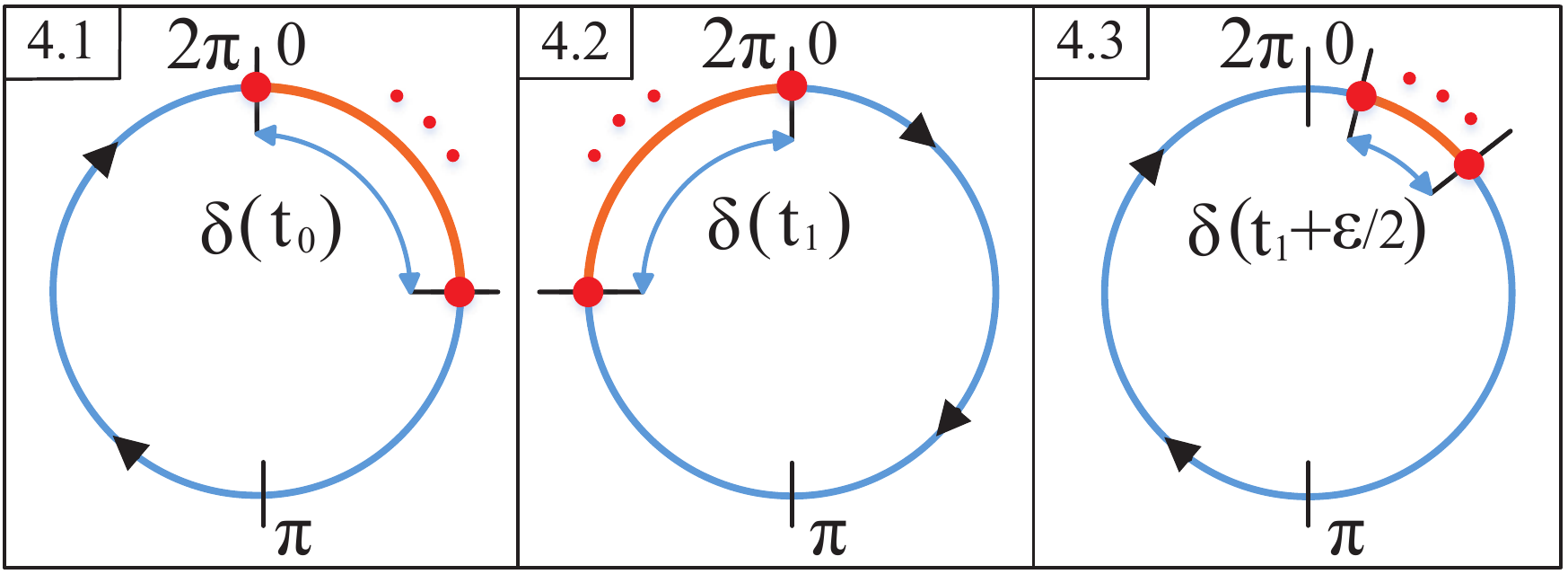}
\caption{Phase distributions of all oscillators at different time instants in \emph{Scenario 1.1}.}
\label{Theorem_1_1_1}
\end{figure}	
		
\emph{Case $1$~($0\leq\epsilon<\pi/2$)}: If $\epsilon$ is $0$, the network is synchronized. So we only consider $0<\epsilon<\pi/2$. Noting that the ending and starting points of the containing arc reside on phases $0$ and $0<\epsilon<\pi/2$ rad at time instant $t_0$, respectively (as depicted in Fig. \ref{Theorem_1_1_1}.1), so after $t_0$, all oscillators will evolve freely without firing for exactly $T-\epsilon>3T/4$ seconds before the starting point of the containing arc reaches phase $2\pi$ rad at time $t_1=t_0+T-\epsilon$ (as depicted in Fig. \ref{Theorem_1_1_1}.2). Meanwhile, the ending point of the containing arc resides on phase $2\pi-\epsilon$ rad and we have $\delta(t_1)=\delta(t_0)=\epsilon$.

Given that the PRF in (\ref{PRF}) is non-negative on $[2\pi-\epsilon,\,2\pi]$, a pulse can only trigger a forward jump or have no effect on an oscillator with phase residing in $[2\pi-\epsilon,\,2\pi]$. So all oscillators will reach phase $2\pi$ rad and fire no later than $t_1+\epsilon$ and within $[t_1,\,t_1+\epsilon/2]$, we can only have one of the following three scenarios:
\leftmargini=22mm
\begin{enumerate}
	\item[\emph{Scenario 1.1:}] all oscillators fired within $[t_1,\,t_1+\epsilon/2]$;
	
	\item[\emph{Scenario 1.2:}] some oscillators did not fire within $[t_1,\,t_1+\epsilon/2]$ but all these oscillators jumped in phase within $[t_1,\,t_1+\epsilon/2]$;
	
	\item[\emph{Scenario 1.3:}] some oscillators neither fired nor jumped in phase within $[t_1,\,t_1+\epsilon/2]$.
\end{enumerate}	

Next, we prove $\delta(t_1+\epsilon)\leq(1-l/2)\delta(t_1)$ in all above three scenarios, based on which we can further prove such a decrease of containing arc after each round of firing and hence the convergence of $\delta(t)$ to zero. Without loss of generality, we label all oscillators in an increasing order of their phases at time instant $t_1$, i.e., $2\pi-\epsilon=\phi_1(t_1)\leq\phi_2(t_1)\leq\cdots\leq\phi_{N}(t_1)=2\pi$ and denote $\mathcal{N}_{f}$ (respectively $\mathcal{N}_{n}$) as the index set of oscillators fired (respectively did not fire) in $[t_1,\,t_1+\epsilon/2]$.

\emph{Scenario 1.1 (all oscillators fired within $[t_1,\,t_1+\epsilon/2]$)}: One can easily know that in this case $\mathcal{N}_{f}$ contains all oscillators and $\mathcal{N}_{n}$ is an empty set. The phases of all oscillators at $t_1+\epsilon/2$ should follow the pattern depicted in Fig. \ref{Theorem_1_1_1}.3.

Since the PRF in (\ref{PRF}) is non-positive on $[0,\,\pi]$, the phase evolution of an oscillator cannot be advanced by received pulses when its phase resides in $[0,\,\pi]$. So all oscillators' phases reside in $[0,\,\epsilon/2]$ at time $t_1+\epsilon/2$, which means $0\leq\delta(t_1+\epsilon/2)\leq\epsilon/2=\delta(t_1)/2$. Given $l\in(0,\,1]$, one can obtain $\delta(t_1+\epsilon/2)\leq(1-l/2)\delta(t_1)$. According to the non-increasing property of the containing arc in Lemma \ref{Lemma_1}, we have $\delta(t_1+\epsilon)\leq(1-l/2)\delta(t_1)$.

\begin{figure}[htbp]
	\centering
	\includegraphics[width=0.43\textwidth]{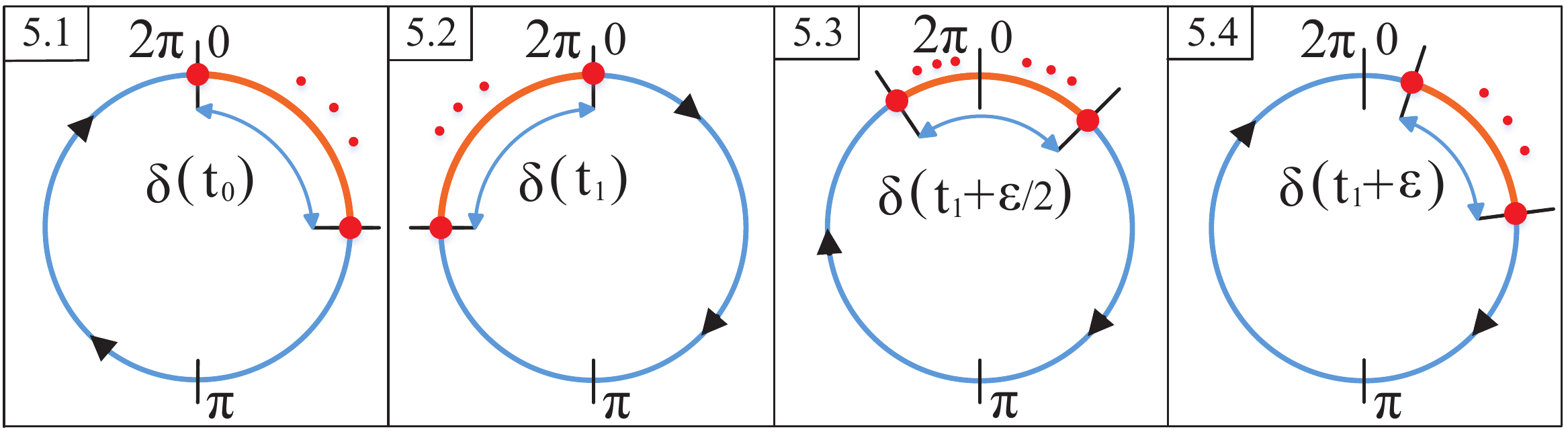}
	\caption{Phase distributions of all oscillators at different time instants in \emph{Scenario 1.2} and \emph{Scenario 1.3}.}
	\label{Theorem_1_1_2}
\end{figure}

\emph{Scenario 1.2 (some oscillators did not fire within $[t_1,\,t_1+\epsilon/2]$ but all these oscillators jumped in phase within $[t_1,\,t_1+\epsilon/2]$)}: At time instant $t_1+\epsilon/2$, the phase distribution of all oscillators should follow the pattern depicted in Fig. \ref{Theorem_1_1_2}.3. The length of the containing arc at $t_1+\epsilon/2$ can be obtained as
\begin{flalign}\label{scenario_1.2.1}
\delta(t_1+\epsilon/2)=\max_{i\in\mathcal{N}_{f}}\{\phi_i(t_1+\epsilon/2)\}+2\pi-\min_{j\in\mathcal{N}_{n}}\{\phi_j(t_1+\epsilon/2)\}
\end{flalign}
Following the same line of reasoning as in \emph{Scenario 1.1}, one can get $\phi_i(t_1+\epsilon/2)\in[0,\,\epsilon/2]$ for $i\in\mathcal{N}_{f}$, i.e.,
\begin{flalign}\label{scenario_1.2.2}
\max_{i\in\mathcal{N}_{f}}\{\phi_i(t_1+\epsilon/2)\}\leq\epsilon/2
\end{flalign}
Next, we characterize $\min_{j\in\mathcal{N}_{n}}\{\phi_j(t_1+\epsilon/2)\}$. Since all oscillators in $\mathcal{N}_{n}$ jumped at least once within $[t_1,\,t_1+\epsilon/2]$, we denote $\hat{t}_j\in[t_1,\,t_1+\epsilon/2]$ as the time instant of oscillator $j$'s first jump within $[t_1,\,t_1+\epsilon/2]$. So the phase of oscillator $j$ immediately before the jump at $\hat{t}_j$ is $\phi_j(\hat{t}_j)=\phi_j(t_1)+\hat{t}_j-t_1$. According to the PRF in (\ref{PRF}), we have the phase of oscillator $j$ immediately after the jump at $\hat{t}_j$ as
\begin{flalign}
\phi_j^+(\hat{t}_j)&=\phi_j(\hat{t}_j)+(2\pi-\phi_j(\hat{t}_j))l\nonumber\\
&=2\pi l+(1-l)(\phi_j(t_1)+\hat{t}_j-t_1)\nonumber
\end{flalign}

Noting that the PRF in (\ref{PRF}) is non-negative on $[2\pi-\epsilon,\,2\pi]$ and oscillator $j$ can be triggered to jump multiple times within $[t_1,\,t_1+\epsilon/2]$, the phase of oscillator $j$ at $t_1+\epsilon/2$ satisfies
\begin{flalign}
\phi_j(t_1+\epsilon/2)&\geq\phi_j^+(\hat{t}_j)+t_1+\epsilon/2-\hat{t}_j\nonumber\\
&=2\pi l+(1-l)\phi_j(t_1)+\epsilon/2-(\hat{t}_j-t_1)l\nonumber
\end{flalign}
Using the facts $\phi_j(t_1)\in[2\pi-\epsilon,\,2\pi]$ and $\hat{t}_j\in[t_1,\,t_1+\epsilon/2]$, we have $\phi_j(t_1+\epsilon/2)\geq2\pi-(1-l)\epsilon/2$ for $j\in\mathcal{N}_{n}$, i.e.,
\begin{flalign}\label{scenario_1.2.4}
\min_{j\in\mathcal{N}_{n}}\{\phi_j(t_1+\epsilon/2)\}\geq2\pi-(1-l)\epsilon/2 
\end{flalign}

Combining (\ref{scenario_1.2.1}), (\ref{scenario_1.2.2}), and (\ref{scenario_1.2.4}), we have $\delta(t_1+\epsilon/2)\leq(1-l/2)\delta(t_1)$. According to the non-increasing property of the containing arc in Lemma \ref{Lemma_1}, one can obtain $\delta(t_1+\epsilon)\leq(1-l/2)\delta(t_1)$.

\emph{Scenario 1.3 (some oscillators neither fired nor jumped in phase within $[t_1,\,t_1+\epsilon/2]$)}: At time instant $t_1+\epsilon/2$, the phase distribution of all oscillators should also follow the pattern depicted in Fig. \ref{Theorem_1_1_2}.3. To prove $\delta(t_1+\epsilon)\leq(1-l/2)\delta(t_1)$, we first characterize the number of oscillators in $\mathcal{N}_{f}$ and $\mathcal{N}_{n}$.

We assume oscillator $j'\in\mathcal{N}_{n}$ neither fired nor jumped in phase within $[t_1,\,t_1+\epsilon/2]$. Recall that no oscillators fired in $(t_0,\,t_1)$ of duration $t_1-t_0=T-\epsilon>3T/4$, according to Mechanism $1$, oscillator $j'$ being not triggered to jump in phase within $[t_1,\,t_1+\epsilon/2]$ implies it receiving no greater than $\lambda_{j'}$ pulses within $[t_1,\,t_1+\epsilon/2]$ of duration less than $T/4$, i.e., condition $b)$ of Mechanism $1$ is not satisfied.

As all oscillators will reach $2\pi$ rad and fire within $[t_1,\,t_1+\epsilon]$, every oscillator $k$ ($1\leq k\leq N$) should receive at least $d(k)$ pulses within $[t_1,\,t_1+\epsilon]$. Since oscillator $j'$ was not triggered to jump and hence received no greater than $\lambda_{j'}$ pulses within $[t_1,\,t_1+\epsilon/2]$, it will receive at least $d(j')-\lambda_{j'}$ pulses in $(t_1+\epsilon/2,\,t_1+\epsilon]$, i.e., the number of oscillators that did not fire in $[t_1,\,t_1+\epsilon/2]$ is at least $d(j')-\lambda_{j'}$. In other words, the number of oscillators in $\mathcal{N}_{n}$ is at least $d(j')-\lambda_{j'}$. According to the definition of $\lambda_{j'}$ in (\ref{lower_bound}), we have $4\lambda_{j'}\leq d(j')-\lfloor N/2\rfloor$, which further leads to $d(j')-\lambda_{j'}\geq \lfloor N/2\rfloor+3\lambda_{j'}$. Given $\lambda_{j'}\geq0$ and $d(j')>\lfloor N/2\rfloor$, we always have $d(j')-\lambda_{j'}\geq\lfloor N/2\rfloor+1$. Therefore, the number of oscillators in $\mathcal{N}_{n}$ is at least $\lfloor N/2\rfloor+1$ and the number of oscillators in $\mathcal{N}_{f}$ is at most $N-(\lfloor N/2\rfloor+1)$, which is no greater than $\lfloor N/2\rfloor$.

Next, we characterize the phases of oscillators at $t_1+\epsilon$. Since all oscillators in $\mathcal{N}_{n}$ fired within $(t_1+\epsilon/2,\,t_1+\epsilon]$, following the same line of reasoning as in \emph{Scenario 1.1}, we have
\begin{flalign}\label{scenario_1.3.1}
\phi_{j}(t_1+\epsilon)\in[0,\,\epsilon/2]
\end{flalign}
for $j\in\mathcal{N}_{n}$.

To determine $\phi_{i}(t_1+\epsilon)$ for $i\in\mathcal{N}_{f}$, we first determine $\phi_{i}(t_1+\epsilon/2)$ for $i\in\mathcal{N}_{f}$. Recall that all oscillators in $\mathcal{N}_{f}$ fired within $[t_1,\,t_1+\epsilon/2]$, following the same line of reasoning as in \emph{Scenario 1.1}, we have $\phi_i(t_1+\epsilon/2)\in[0,\,\epsilon/2]$ for $i\in\mathcal{N}_{f}$. Next, we prove that all oscillators in $\mathcal{N}_{f}$ will be triggered to jump in phase within $(t_1+\epsilon/2,\,t_1+\epsilon]$.

As has been proven, the number of oscillators in $\mathcal{N}_{f}$ is no greater than $\lfloor N/2\rfloor$ and all oscillators in $\mathcal{N}_{f}$ fired within $[t_1,\,t_1+\epsilon/2]$. So every oscillator $i$ in $\mathcal{N}_{f}$ can receive at most $\lfloor N/2\rfloor-1$ pulses within $[t_1,\,t_1+\epsilon/2]$ (note that oscillator $i$ cannot receive its own pulse) and will receive at least $d(i)-(\lfloor N/2\rfloor-1)$ pulses within $(t_1+\epsilon/2,\,t_1+\epsilon]$ of duration less than $T/4$. Using the definition of $\lambda_i$ in (\ref{lower_bound}), we have $d(i)-(\lfloor N/2\rfloor-1)>\lambda_i$, i.e., there must exist a time instant $\tilde{t}_i\in (t_1+\epsilon/2,\,t_1+\epsilon]$ for every oscillator $i$ at which it receives the $(\lambda_i+1)th$ pulse since (but not including) time instant $t_1+\epsilon/2$, i.e., condition $b)$ in Mechanism $1$ is satisfied. Next we proceed to prove that at $\tilde{t}_i$, condition $c)$ in Mechanism $1$ is also satisfied (note that condition $a)$ is always satisfied since we start at $t_0>2T$), and hence all oscillators in $\mathcal{N}_{f}$ will be triggered to jump in phase in $(t_1+\epsilon/2,\,t_1+\epsilon]$.

As no oscillators fire within $(t_0,\,t_1)$ of duration $t_1-t_0=T-\epsilon>3T/4$ and oscillator $i$ receives at most $\lfloor N/2\rfloor-1$ pulses within $[t_1,\,t_1+\epsilon/2]$, we have that within $(t_0,\,t_1+\epsilon/2]$ of duration $t_1+\epsilon/2-t_0>3T/4$, oscillator $i$ receives at most $\lfloor N/2\rfloor-1$ pulses, which is less than $\bar{\lambda}_i-2\lambda_i$ according to (\ref{upper_bound}), implying that at $\tilde{t}_i$, condition $c)$ of Mechanism $1$ is also satisfied. Therefore, according to Mechanism $1$, the phase of oscillator $i$ will be triggered to jump by the pulse received at $\tilde{t}_i$, i.e., every oscillator $i$ in $\mathcal{N}_{f}$ will be triggered to jump in phase within $(t_1+\epsilon/2,\,t_1+\epsilon]$.

Now we are in position to determine the phase of oscillator $i$ for $i\in\mathcal{N}_{f}$ at time instant $t_1+\epsilon$. Since every oscillator $i$ jumped at least once within $(t_1+\epsilon/2,\,t_1+\epsilon]$, we denote $\hat{t}_i\in(t_1+\epsilon/2,\,t_1+\epsilon]$ as the time instant of oscillator $i$'s first jump within $(t_1+\epsilon/2,\,t_1+\epsilon]$. So the phase of oscillator $i$ immediately before the jump at $\hat{t}_i$ is $\phi_i(\hat{t}_i)=\phi_i(t_1+\epsilon/2)+\hat{t}_i-(t_1+\epsilon/2)$. According to the PRF in (\ref{PRF}), the phase of oscillator $i$ immediately after the jump at $\hat{t}_i$ can be obtained as
\begin{flalign}
\phi_i^+(\hat{t}_i)=&(1-l)\phi_i(\hat{t}_i)\nonumber\\
=&(1-l)(\phi_i(t_1+\epsilon/2)+\hat{t}_i-(t_1+\epsilon/2))\nonumber
\end{flalign}

Noting that the PRF in (\ref{PRF}) is non-positive on $[0,\,\pi]$ and oscillator $i$ can be triggered to jump multiple times within $(t_1+\epsilon/2,\,t_1+\epsilon]$, the phase of oscillator $i$ at $t_1+\epsilon$ satisfies
\begin{flalign}\label{scenario_1.3.2}
\phi_i(t_1+\epsilon)&\leq\phi_i^+(\hat{t}_i)+(t_1+\epsilon)-\hat{t}_i\nonumber\\
&\leq (1-l)\phi_j(t_1+\epsilon/2)+\epsilon/2+(t_1+\epsilon/2-\hat{t}_i)l
\end{flalign}

Substituting $\phi_i(t_1+\epsilon/2)\in[0,\,\epsilon/2]$ and $\hat{t}_i\in(t_1+\epsilon/2,\,t_1+\epsilon]$ into (\ref{scenario_1.3.2}) leads to $\phi_i(t_1+\epsilon)\in[0,\,(1-l/2)\epsilon]$ for $i\in\mathcal{N}_{f}$. In combination with the fact $\phi_{j}(t_1+\epsilon)\in[0,\,\epsilon/2]$ for $j\in\mathcal{N}_{n}$ in (\ref{scenario_1.3.1}) and $l\in(0,\,1]$, we have that the phases of all oscillators reside in $[0,\,(1-l/2)\epsilon]$ at time $t_1+\epsilon$, i.e., $\delta(t_1+\epsilon)\leq(1-l/2)\delta(t_1)$.

In summary, we have $\delta(t_1+\epsilon)\leq(1-l/2)\delta(t_1)$ in all three \emph{Scenarios} \emph{1.1}, \emph{1.2}, and \emph{1.3}. At $t_1+\epsilon$, all oscillators reside in $[0,\pi]$ and will evolve towards phase $2\pi$ rad and fire. By repeating the above analyses, we can get that the length of the containing arc $\delta(t)$ decreases to a value no greater than $(1-l/2)\delta(t)$ after each round of firing until it converges to $0$. Therefore, synchronization can be achieved in Case $1$.

%%%%%%%%%%%%%%%%%%%%%%%%%%%%%%%%%%%%%%%%%%%%%%%%%%%%%%%%%%%%%%%%%%%%%%%%%%%%%%%%%%%%%%%%%%%%%%%
%%%%%%%%%%%%%%%%%%%%%%%%%%%%%%%%%%%%%%%%%%%%%%%%%%%%%%%%%%%%%%%%%%%%%%%%%%%%%%%%%%%%%%%%%%%%%%%

\emph{Case $2$~($\pi/2\leq\epsilon<\pi$)}: Similar to the reasoning in \emph{Case $1$}, there exists a time instant $t_0>2T$ at which the ending and starting points of the containing arc reside on phases $0$ and $\pi/2\leq\epsilon<\pi$ rad, respectively. After $t_0$, all oscillators evolve freely for exactly $T-\epsilon>T/2$ seconds before the starting point of the containing arc reaches phase $2\pi$ rad at $t_1=t_0+T-\epsilon$. At $t_1$, the ending point of the containing arc resides on phase $2\pi-\epsilon$ rad and we have $\delta(t_1)=\delta(t_0)=\epsilon$.

Given that the PRF in (\ref{PRF}) is non-negative on $[2\pi-\epsilon,\,2\pi]$, a pulse can only trigger a forward jump or have no effect on an oscillator with phase residing in $[2\pi-\epsilon,\,2\pi]$. So all oscillators will reach phase $2\pi$ rad and fire no later than time instant $t_1+\epsilon$ and within $[t_1,\,t_1+\epsilon/2]$, only one of the following three scenarios can happen:
\leftmargini=22mm
\begin{enumerate}
	\item[\emph{Scenario 2.1:}] all oscillators fired within $[t_1,\,t_1+\epsilon/2]$;
	
	\item[\emph{Scenario 2.2:}] some oscillators did not fire within $[t_1,\,t_1+\epsilon/2]$ but all of these oscillators jumped in phase within $[t_1,\,t_1+\epsilon/2]$;
	
	\item[\emph{Scenario 2.3:}] some oscillators neither fired nor jumped in phase within $[t_1,\,t_1+\epsilon/2]$.
\end{enumerate}	

Next, we show that $\delta(t)$ will decrease to less than $\pi/2$ rad in finite time, meaning that \emph{Case $2$} will shift to \emph{Case $1$} in finite time. Therefore, $\delta(t)$ will also converge to $0$ for $\pi/2\leq\epsilon<\pi$.

Similar to \emph{Case $1$}, we label all oscillators in an increasing order of their phases at $t_1$, i.e., $2\pi-\epsilon=\phi_1(t_1)\leq\phi_2(t_1)\leq\cdots\leq\phi_{N}(t_1)=2\pi$ and denote $\mathcal{N}_{f}$ (respectively $\mathcal{N}_{n}$) as the index set of oscillators fired (respectively did not fire) in $[t_1,\,t_1+\epsilon/2]$. Following the same line of reasoning as in \emph{Scenario 1.1} and \emph{Scenario 1.2}, one can easily obtain $\delta(t_1+\epsilon)\leq(1-l/2)\delta(t_1)$ in \emph{Scenario 2.1} and \emph{Scenario 2.2}, respectively. For \emph{Scenario 2.3}, i.e., some oscillators neither fired nor jumped in phase within $[t_1,\,t_1+\epsilon/2]$, we assume oscillator $j'$ is such an oscillator. According to Mechanism $1$, there could be two reasons for the not firing of oscillator $j'$ in $[t_1,\,t_1+\epsilon/2]$:

%%%%%%%%%%%%%%%%%%%%%%%%%%%%%%%%%%%%%%%%%%%%%%%%%%%%%%%%%%%%%%%%%%%%%%%%
\leftmargini=23.8mm
\leftmarginii=23.8mm
\begin{enumerate}
	\item[\emph{Scenario 2.3.1:}] oscillator $j'$ receives no greater than $\lambda_{j'}$ pulses within $[t_1,\,t_1+\epsilon/2]$, i.e., condition $b)$ of Mechanism $1$ is not satisfied;
	\item[\emph{Scenario 2.3.2:}] oscillator $j'$ receives over $\lambda_{j'}$ pulses within $[t_1,\,t_1+\epsilon/2]$, but the number of pulses it received within the past period of length $3T/4$ is no less than $\bar\lambda_{j'}$, i.e., condition $c)$ of Mechanism $1$ is not satisfied.
\end{enumerate}
%%%%%%%%%%%%%%%%%%%%%%%%%%%%%%%%%%%%%%%%%%%%%%%%%%%%%%%%%%%%%%%%%%%%%%%%

Next, we show that in both scenarios, the length of the containing arc will keep decreasing to less than $(1-l/2)$ of its original value.

\emph{Scenario 2.3.1}: Following the same line of reasoning as in \emph{Scenario 1.3}, all oscillators' phases reside in $[0,\,(1-l/2)\epsilon]$ at time instant $t_1+\epsilon$, which means $\delta(t_1+\epsilon)\leq(1-l/2)\delta(t_1)$.

\begin{figure}[htbp]
	\centering
	\includegraphics[width=0.43\textwidth]{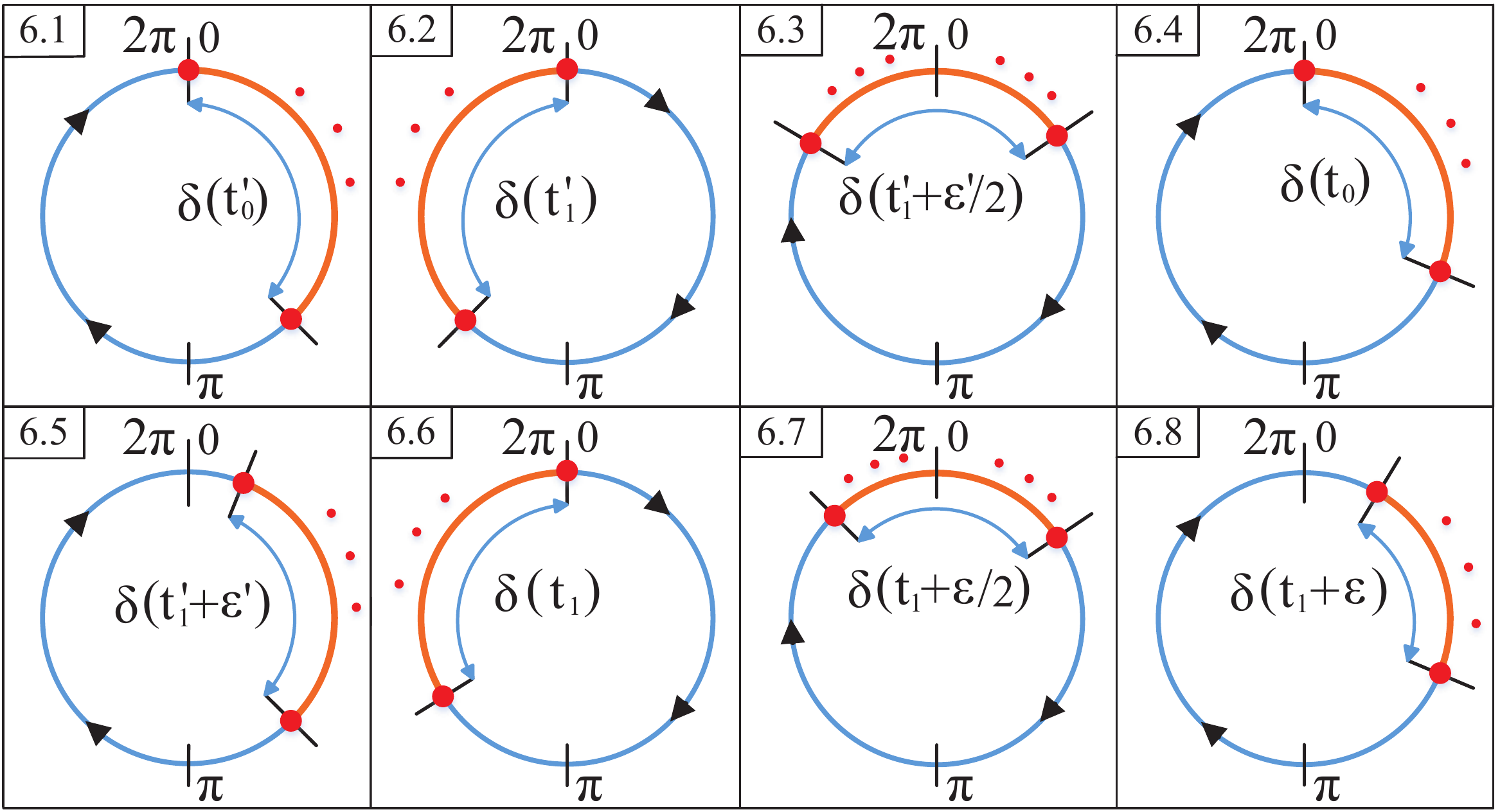}
	\caption{Phase distributions of all oscillators at different time instants in \emph{Scenario 2.3.2}.}	
	\label{Theorem_1_2}
\end{figure}

\emph{Scenario 2.3.2}: In this case, we cannot prove length decrease in the containing arc by focusing on the time interval $[t_0,\,t_1+\epsilon]$ (one firing round), so we extend our considered time span to two firing rounds. Without loss of generality, we assume that the previous firing round starts at $t_0'<t_0$ at which the ending and starting points of the containing arc reside on phases $0$ and $\epsilon'$ rad, respectively (as depicted in Fig. \ref{Theorem_1_2}.1). As the containing arc is non-increasing (Lemma \ref{Lemma_1}), we have $\epsilon\leq\delta(t_0')=\epsilon'<\pi$. After $t_0'$, all oscillators evolve freely for exactly $2\pi-\epsilon'>T/2$ seconds before the starting point of the containing arc reaches phase $2\pi$ rad at time $t_1'=t_0'+2\pi-\epsilon'$ (as depicted in Fig. \ref{Theorem_1_2}.2). At $t_1'$, the ending point of the containing arc resides on phase $2\pi-\epsilon'$ rad and we have $\delta(t_1')=\delta(t_0')=\epsilon'$.

Given that the PRF in (\ref{PRF}) is non-negative on $[2\pi-\epsilon',\,2\pi]$, a pulse can only trigger a forward jump or have no effect on an oscillator with phase residing in $[2\pi-\epsilon',\,2\pi]$. So all oscillators will reach phase $2\pi$ rad and fire no later than $t_1'+\epsilon'$. The phases of all oscillators at $t_1'+\epsilon'$ should follow the pattern depicted in Fig. \ref{Theorem_1_2}.5. Next, we prove $\delta(t_1'+\epsilon')\leq(1-l/2)\delta(t_1')$. To this end, we need to characterize the number of oscillators fired within $[t_1',\,t_1'+\epsilon'/2]$. The phases of all oscillators follow the pattern depicted in Fig. \ref{Theorem_1_2}.3 at time instant $t_1'+\epsilon'/2$. We denote $\mathcal{N}_{f}'$ (respectively $\mathcal{N}_{n}'$) as the index set of oscillators fired (respectively did not fire) within $[t_1',\,t_1'+\epsilon'/2]$ and analyze the numbers of oscillators in the two sets.

Recall that in \emph{Scenario 2.3.2}, condition $c)$ of Mechanism $1$ is not satisfied. So oscillator $j'$ should receive at least $\bar{\lambda}_{j'}-\lambda_{j'}$ pulses within $(t_1-3T/4,\,t_1)$. Since no oscillators fired within $(t_0,\,t_1)$, the number of oscillators fired in $(t_1-3T/4,\,t_0]$ is at least $\bar{\lambda}_{j'}-\lambda_{j'}$. Next, by proving $(t_1-3T/4,\,t_0]\subseteq(t_1'+\epsilon'/2,\,t_1'+\epsilon']$, we show that the number of oscillators fired in $(t_1'+\epsilon'/2,\,t_1'+\epsilon']$ is no less than $\bar{\lambda}_{j'}-\lambda_{j'}$. As indicated earlier, all oscillators will reach phase $2\pi$ rad and fire no later than $t_1'+\epsilon'$. So we have $t_0\leq t_1'+\epsilon'$. On the other hand, since the starting point of the containing arc resides on phase $\pi/2\leq\epsilon<\pi$ at $t_0$ and the PRF in (\ref{PRF}) is non-positive on $[0,\,\epsilon]$, oscillators having phase in $[0,\,\epsilon]$ will not be advanced by incoming pulses. So it takes an oscillator at least $\epsilon$ time to evolve from $0$ to $\epsilon$ rad. Therefore, we can obtain $t_0-t_1'\geq\epsilon$. Given $\epsilon'<\pi=T/2$ and $t_1=t_0+T-\epsilon$, one can get
\[
t_1'+\epsilon'/2\leq t_0-\epsilon+\epsilon'/2<t_0-\epsilon+T/4=t_1-3T/4
\]
and hence $(t_1-3T/4,\,t_0]\subseteq(t_1'+\epsilon'/2,\,t_1'+\epsilon']$, implying that at least $\bar{\lambda}_{j'}-\lambda_{j'}$ oscillators fired within $(t_1'+\epsilon'/2,\,t_1'+\epsilon']$. According to the definition of $\lambda_{j'}$ and $\bar{\lambda}_{j'}$ in (\ref{lower_bound}) and (\ref{upper_bound}), we have $4\lambda_{j'}\leq d(j')-\lfloor N/2\rfloor$ and $\bar{\lambda}_{j'}-\lambda_{j'}=d(j')-3\lambda_{j'}$, which further lead to $d(j')-3\lambda_{j'}\geq \lfloor N/2\rfloor+\lambda_{j'}$. Given $\lambda_{j'}\geq0$ and $d(j')>\lfloor N/2\rfloor$, we always have $d(j')-3\lambda_{j'}\geq\lfloor N/2\rfloor+1$. Therefore, the number of oscillators in ${\mathcal{N}}_{n}'$ is at least $\lfloor N/2\rfloor+1$ and the number of oscillators in ${\mathcal{N}}_{f}'$ is at most $N-(\lfloor N/2\rfloor+1)$, which is no greater than $\lfloor N/2\rfloor$.

Based on obtained knowledge of the numbers of oscillators in $\mathcal{N}_{f}'$ and $\mathcal{N}_{n}'$, respectively, we can characterize the phases of all oscillators at time instant $t_1'+\epsilon'$. Following the same line of reasoning as in \emph{Scenario 1.3}, one can obtain that all oscillators' phases reside in $[0,\,(1-l/2)\epsilon']$ at time instant $t_1'+\epsilon'$, which means $\delta(t_1'+\epsilon')\leq(1-l/2)\delta(t_1')$. Note that proving such a length decrease of the containing arc requires a careful characterization of phase evolution starting from $t_0'$ to $t_1+\epsilon$, which spans two consecutive firing rounds. After $t_1+\epsilon$, the phase evolution could follow \emph{Scenario 2.1}, \emph{Scenario 2.2}, \emph{Scenario 2.3.1} (in which we can prove such $(1-l/2)$ length decrease after each round of firing) or \emph{Scenario 2.3.2} (in which we can prove such $(1-l/2)$ length decrease after every two consecutive firing rounds).

In summary, we can prove that the length of the containing arc will reduce to $(1-l/2)$ of its original value after every firing round in \emph{Scenarios} \emph{2.1}, \emph{2.2}, and \emph{2.3.1}, whereas in \emph{Scenario 2.3.2}, we can prove such a decrease after every two consecutive firing rounds. Since every oscillator will fire at least once within any time interval of length $3T/2$ according to Lemma \ref{Lemma_2}, we can get that the length of the containing arc $\delta(t)$ will decrease to a value less than $\pi/2$ rad within finite time (in fact, after at most $2m$ firing rounds with $m$ satisfying $(1-l/2)^{m}\delta(t_0)<\pi/2$). And then, the containing arc will keep decreasing to $0$ following the derivations in \emph{Case $1$}.

By combining \emph{Case $1$} and \emph{Case $2$}, one can obtain that $\delta(t)$ will always converge to $0$ under the conditions of Theorem \ref{Theorem_1_WithoutAttack}.
\end{proof}

\begin{Corollary}\label{Corollary_1_2pi_period}
Under conditions in Theorem \ref{Theorem_1_WithoutAttack}, Mechanism $1$ guarantees that all oscillators synchronize with an oscillation period $T=2\pi$ seconds in the absence of attacks.
\end{Corollary}
\begin{proof}The result can be easily obtained from the reasoning in the proof of Theorem \ref{Theorem_1_WithoutAttack} and hence is omitted.
\end{proof}

\begin{Remark}
Besides enabling attack resilience, Mechanism $1$ also has better robustness against time-varying delays. For example, numerical simulations in Fig. \ref{Sync_error_1} and Fig. \ref{Sync_error_2} show that Mechanism $1$ has much smaller synchronization errors compared with synchronization mechanisms in \cite{klinglmayr2012self,yun2015robustness,wang2018pulse} when the communication is subject to random time-varying delays.
\end{Remark}

%%%%%%%%%%%%%%%%%%%%%%%%%%%%%%%%%%%%%%%%%%%%%%%%%%%%%%%%%%%%%%%%%%%%%%%%%%%%%%%%

\section{Stealthy Byzantine Attacks and Attack Detection Mechanism}

The concept of Byzantine attacks stems from the Byzantine generals problem \cite{lamport1982byzantine}. It was used to describe a traitor commander who sends or relays fake information to other commanders to avoid the loyal ones from reaching agreement \cite{pease1980reaching}. In the case of PCO synchronization, a node compromised by Byzantine attacks can emit malicious pulses at arbitrary time instants. However, given that the purpose of Byzantine attacks is to delay or damage the synchronization of legitimate oscillators, we assume that a compromised oscillator sends malicious pulses only when such pulses can negatively affect the synchronization process of legitimate oscillators, i.e., enlarge the containing arc of affected legitimate oscillators. 

A compromised node decides the timing of its malicious pulses based on information of other oscillator's phases that it can perceive from received pulses. Given that in a general connected PCO network, an oscillator can only receive pulses from its neighbors, a compromised oscillator can only perceive phase information of nodes that it can receive pulses from and decide its optimal attacking strategy accordingly. 

We consider two types of attacks, non-colluding attacks and colluding attacks. In non-colluding attacks, an attacker determines its attacking strategy based on its own neighbors' phase information. In colluding attacks, two attackers can share perceived phase information about each other's neighbors, which is equivalent to expanding the neighbor sets of both attackers to the union of their neighbor sets. The same concept can be extended to three or more colluding attackers.  

Now we proceed to discuss the attacking strategy. If an attacker keeps sending pulses continuously without rest, it can effectively prevent legitimate oscillators from reaching synchronization. However, such attacks are not energy efficient and will also render themselves easily detectable, just as jamming of communication channels being easy to detect, isolate, and remove \cite{xu2005feasibility}. Therefore, we are only interested in ``stealthy" Byzantine attacks, in which attack pulses are emitted in a way that cannot be detected by legitimate oscillators in the pulse-based interaction framework. 

In PCO networks, since all exchanged pulses are identical without embedded content such as source or destination information, conventional content-checking based attack-detection mechanisms such as \cite{lamport1985synchronizing} are inapplicable. We propose to let each oscillator detect potential attacks by monitoring the number of pulses it receives within a certain time interval. The basic rationale is as follows: In a given time interval, if the number of received pulses is greater than the maximally possible number of pulses emitted by all legitimate oscillators, then it is safe to conclude that an attacker is present who injected the superfluous pulses. To this end, we first characterize the number of pulses that an oscillator can receive within a certain time interval in the absence of attacks.
\begin{Lemma}\label{Lemma_3_bound}
For a general connected network of $N$ legitimate PCOs, under Mechanism $1$, an oscillator $i$ can receive at most $d^-(i)$ pulses within any time interval $[t,\,t+T/2]$ for $t\geq0$ where $d^-(i)$ is the indegree of oscillator $i$.
\end{Lemma}
\begin{proof} Noting that the number of edges entering oscillator $i$ is $d^-(i)$ in the considered general connected PCO network, Lemma \ref{Lemma_3_bound} can be obtained following the same line of reasoning as in Theorem $3$ of \cite{wang2018pulse}.
\end{proof}

Based on Lemma \ref{Lemma_3_bound}, we have, under the pulse-number based detection mechanism, that oscillator $i$'s receiving more than $d^-(i)$ pulses within an arbitrary time interval $[t,\,t+T/2]$ implies the presence of attackers among its neighbors. Therefore, to keep stealthy, one compromised oscillator should launch stealthy attacks by sending pulses with a time separation over $T/2$ seconds.

From the above analysis, we summarize the attacking models as follows:

In non-colluding attacks, a Byzantine attacker emits an attack pulse only when the pulse can enlarge the containing arc of its neighbors. In addition, to keep stealthy, every individual attacker sends malicious pulses with a time separation over $T/2$ seconds.

In colluding attacks, a Byzantine attacker emits an attack pulse either when the pulse  can enlarge the containing arc of the union set of colluding attackers' neighbor sets, or  when the pulse can help other attack pulse to do so.

\section{Synchronization of PCO Networks under Stealthy Byzantine Attacks}

In this section, we address the synchronization of general connected PCO networks in the presence of stealthy Byzantine attacks. Among $N$ PCOs, we assume that $M$ are compromised and act as stealthy Byzantine attackers. We first show that the proposed pulse-based synchronization mechanism (Mechanism $1$) can synchronize legitimate oscillators when attackers do not collude, i.e., every attacker determines its attacking strategy based on its own neighbors' phase information. Then we further prove that all legitimate oscillators can still be synchronized even when attackers collude with each other, i.e., attackers can exchange phase information of their neighbors. To this end, we first analyze the phase evolution of legitimate oscillators in the presence of non-colluding attackers.

\begin{Lemma}\label{Lemma_4_encounter_attack}
For a general connected network of $N$ PCOs, within which $M\leq2\times\lfloor(d-\lfloor N/2\rfloor )/4\rfloor$ oscillators are compromised non-colluding attackers launching attacks following the stealthy Byzantine attack model in Section \uppercase\expandafter{\romannumeral4}, if the initial length of the containing arc of legitimate oscillators is less than $\pi$ and $d>\lfloor N/2 \rfloor$, then under Mechanism $1$, the $N-M$ legitimate oscillators encounter attack pulses only when their phases reside partially in $[0,\,\pi)$, partially in $(\pi, 2\pi]$ with phase $2\pi$ belonging to the containing arc.
\end{Lemma}

\begin{proof}
According to Mechanism $1$, all legitimate oscillators will evolve freely for an entire period $T=2\pi$. Since the initial length of the containing arc is assumed to be less than $\pi$, the possible phase distribution of all legitimate oscillators immediately after the initial period of free evolution can only fall within one of the following four scenarios, as depicted in Fig. \ref{keep_firing}:

\begin{itemize}
\item[I)] all legitimate oscillators' phases reside in $[0,\,\pi]$;
\item[II)] legitimate oscillators' phases reside partially in $(0,\,\pi]$, partially in $(\pi,\,2\pi]$ with phase $\pi$ belonging to the containing arc;
\item[III)] all legitimate oscillators' phases reside in $(\pi,\,2\pi]$;
\item[IV)] legitimate oscillators' phases reside partially in $[0,\,\pi)$, partially in $(\pi,\,2\pi]$ with phase $2\pi$ belonging to the containing arc.
\end{itemize}

Since in non-colluding attacks, an attacker will emit an attack pulse only when the pulse can enlarge the containing arc of its legitimate neighbors, every attack pulse will trigger a phase shift on at least one legitimate oscillator. Next, we prove that an attacker can trigger a legitimate oscillator (say oscillator $j$) to jump in phase only under Scenario \uppercase\expandafter{\romannumeral4}).

\begin{itemize}
\item[I)] All legitimate oscillators' phases reside in $[0,\,\pi]$. Without loss of generality, we assume that legitimate oscillator $k$ fires last among all legitimate oscillators at time instant $t_k$. One can easily get that all legitimate oscillators fired in the past $T/2$ seconds prior to $t_k$. Recalling $d\triangleq\min_{i=1,2,\cdots,N}\{d(i)\}$, we have $M\leq2\times\lfloor(d-\lfloor N/2\rfloor )/4\rfloor\leq2\times\lfloor(d(i)-\lfloor N/2\rfloor )/4\rfloor=2\lambda_i$. Hence, immediately after the firing of oscillator $k$, legitimate oscillator $i$ has received at least $d(i)-M\geq\bar{\lambda}_i$ legitimate pulses during $[t_k-T/2,\,t_k]$ for $i\in\mathcal{N}_L$ where $\mathcal{N}_L$ is the index set of all legitimate oscillators. According to Mechanism $1$, if legitimate oscillator $i$ received no less than $\bar{\lambda}_i$ pulses within the past $3T/4$, no pulse can trigger oscillator $i$ to jump in phase. Hence, immediately after the firing of legitimate oscillator $k$, all legitimate oscillators will evolve freely for $T/4$ and no pulses can trigger a legitimate oscillator to jump in phase within this period. After this quarter period, legitimate oscillators will not emit pulses before the network shifts to Scenario \uppercase\expandafter{\romannumeral2}) and the number of attacker pulses is not enough to trigger a legitimate oscillator to jump in phase. Given that an attacker sends pulses only when the containing arc of its legitimate neighbors can be enlarged, no attack pulse will be emitted in this scenario.

\item[II)] Legitimate oscillators' phases reside partially in $(0,\,\pi]$, partially in $(\pi,\,2\pi]$ with phase $\pi$ belonging to the containing arc. Following the same line of reasoning as in Scenario \uppercase\expandafter{\romannumeral1}), one can get that no legitimate oscillators reach phase $2\pi$ and fire in this scenario. Because no attack pulse can shift the phase of a legitimate oscillator, no attacker will emit attack pulses in this scenario.

\item[III)] All legitimate oscillators' phases reside in $(\pi,\,2\pi]$. One can get that no legitimate oscillators fire in the past $T/4$. Since the number of attacker pulses is not enough to trigger a legitimate oscillator to jump in phase, no attacker will emit attack pulses in this scenario.

\item[IV)] Legitimate oscillators' phases reside partially in $[0,\,\pi)$, partially in $(\pi,\,2\pi]$ with phase $2\pi$ belonging to the containing arc. One can get that a portion of legitimate oscillators fired in the past $T/4$ in this scenario. So an attacker may be able to emit an attack pulse at a right time instant to trigger legitimate neighbors to jump in phase and enlarge the containing arc of its legitimate neighbors.
\end{itemize}

By iterating the above analysis, we can get that an attacker will emit an attack pulse to shift the phase of a legitimate oscillator only when legitimate oscillators' phases reside partially in $[0, \pi)$, partially in $(\pi, 2\pi]$ with phase $2\pi$ rad belonging to the containing arc.
\end{proof}

Next, we establish the synchronization condition for general connected PCO networks in the presence of non-colluding stealthy Byzantine attackers.

\begin{Theorem}\label{Theorem_2_WithAttacks_Noclue}
For a general connected network of $N$ PCOs, within which $M\leq2\times\lfloor (d-\lfloor N/2\rfloor)/4\rfloor$ oscillators are compromised non-colluding attackers launching attacks following the stealthy Byzantine attack model in Sec \uppercase\expandafter{\romannumeral4}, if the initial length of the containing arc of legitimate oscillators is less than $\pi$ rad and $d>\lfloor N/2\rfloor$, then the containing arc of legitimate oscillators will converge to zero under Mechanism $1$.
\end{Theorem}

\begin{proof} We divide the proof into two parts. In Part \uppercase\expandafter{\romannumeral1}, we prove that the length of the containing arc of legitimate oscillators is non-increasing. In Part \uppercase\expandafter{\romannumeral2}, we prove that it converges to $0$. 
	
Part \uppercase\expandafter{\romannumeral1} \emph{(The length of the containing arc of legitimate oscillators is non-increasing)}: It can be easily inferred that the length of the containing arc of legitimate oscillators remains unchanged if no legitimate oscillators jump in phase. So we only consider the case where a pulse (from either a legitimate oscillator or an attacker) triggers a phase jump on a legitimate oscillator.

As no legitimate oscillators will be triggered to jump in phase in the first free-running period, we only consider pulses sent after $t=T$. We will show that for any pulse sent at $t_i>T$, the length of the containing arc of legitimate oscillators is non-increasing. 
	
When the pulse is from a legitimate oscillator $i$, we have $\phi_i(t_i)=2\pi$, i.e., at $t_i$ the containing arc of legitimate oscillators includes phase $2\pi$ rad. Following the same line of reasoning as in Lemma \ref{Lemma_1}, one can obtain that the pulse cannot increase the length of the containing arc of legitimate oscillators.
	
When the pulse is from an attacker, according to Lemma \ref{Lemma_4_encounter_attack}, the pulse can only be sent when legitimate oscillators' phases reside partially in $[0, \pi)$, partially in $(\pi, 2\pi]$ with phase $2\pi$ rad belonging to the containing arc. Following the same line of reasoning as in Scenario $c)$ of Lemma \ref{Lemma_1}, one can obtain that the length of the containing arc of all legitimate oscillators cannot be increased by the attack pulse, although the containing arc of a subset of legitimate oscillators (an attacker's neighbor set) will be enlarged, as confirmed later in the numerical simulations in Fig. \ref{Attack_N_known_No_clue}. Hence, we can conclude that the length of the containing arc of all legitimate oscillators is non-increasing.

Part \uppercase\expandafter{\romannumeral2} \emph{(The length of the containing arc of legitimate oscillators converges to $0$)}: First, we prove that every legitimate oscillator will fire at least once within any time interval of length $3T/2$. According to the argument in Lemma \ref{Lemma_4_encounter_attack}, attack pulses will only be emitted when legitimate oscillators' phases reside partially in $[0,\,\pi)$, partially in $(\pi,\,2\pi]$ with phase $2\pi$ rad belonging to the containing arc. Following the same line of reasoning as in Lemma \ref{Lemma_2}, we can easily get that every legitimate oscillator will fire at least once within any time interval of length $3T/2$.

Next, we prove that the length of the containing arc of legitimate oscillators will decrease to $0$. Without loss of generality, we denote $\delta(t)$ as the length of the containing arc of legitimate oscillators at $t$ and set the initial time to $t=0$. According to the argument in Part \uppercase\expandafter{\romannumeral1}, we have that $\delta(t)$ is non-increasing and $0\leq\delta(t)<\pi$ for $t\geq 0$. Since every legitimate oscillator will fire at least once within any time interval of length $3T/2$, there exists a time instant $t_0>2T$ at which the ending point of the containing arc of legitimate oscillators resides at phase $0$. Denoting the starting point of the containing arc at $t_0$ as $0\leq\epsilon<\pi$, we have $\delta(t_0)=\epsilon$. Next, we separately discuss the $0\leq \epsilon<\pi/2$ case and the $\pi/2\leq\epsilon<\pi$ case to prove the convergence of $\delta(t)$ to $0$.

\emph{Case} \uppercase\expandafter{\romannumeral1}~($0\leq\epsilon<\pi/2$): If $\epsilon$ is $0$, the network is synchronized. So we only consider $0<\delta(t_0)<\pi/2$. At time instant $t_0$, the ending and starting points of the containing arc of legitimate oscillators reside on phases $0$ and $0<\epsilon<\pi/2$ rad, respectively. According to Lemma \ref{Lemma_4_encounter_attack}, attack pulses are emitted only when legitimate oscillators' phases reside partially in $[0,\,\pi)$, partially in $(\pi,\,2\pi]$ with phase $2\pi$ rad belonging to the containing arc. So after $t_0$, all legitimate oscillators will evolve freely without perturbation for exactly $T-\epsilon>3T/4$ seconds before the starting point of the containing arc reaches phase $2\pi$ rad at time $t_1=t_0+T-\epsilon$. At $t_1$, the ending point of the containing arc resides on phase $2\pi-\epsilon$ rad and we have $\delta(t_1)=\delta(t_0)=\epsilon$. Given that the PRF in (\ref{PRF}) is non-negative on $[2\pi-\epsilon,\,2\pi]$, a pulse can only trigger a forward jump or have no effect on a legitimate oscillator with phase residing in $[2\pi-\epsilon,\,2\pi]$. All legitimate oscillators will reach phase $2\pi$ rad and fire no later than $t_1+\epsilon$ and within $[t_1,\,t_1+\epsilon/2]$, we can only have one of the following three scenarios:

\leftmargini=22mm
\begin{enumerate}
	\item[\emph{Scenario \uppercase\expandafter{\romannumeral1}.1:}] all legitimate oscillators fired within $[t_1,\,t_1+\epsilon/2]$;
	
	\item[\emph{Scenario \uppercase\expandafter{\romannumeral1}.2:}] some legitimate oscillators did not fire within $[t_1,\,t_1+\epsilon/2]$ but all of these legitimate oscillators jumped in phases within $[t_1,\,t_1+\epsilon/2]$;
	
	\item[\emph{Scenario \uppercase\expandafter{\romannumeral1}.3:}] some legitimate oscillators neither fired nor jumped in phase within $[t_1,\,t_1+\epsilon/2]$.
\end{enumerate}

Next, we prove $\delta(t_1+\epsilon)\leq(1-l/2)\delta(t_1)$ in all above three scenarios, based on which we can further prove such a length decrease of containing arc of legitimate oscillators after each round of firing and hence the convergence of $\delta(t)$ to zero.

Following the same line of reasoning as in \emph{Scenarios 1.1}, \emph{1.2}, and \emph{1.3} of Theorem \ref{Theorem_1_WithoutAttack} and using the fact that the number of attackers $M$ is no greater than $2\times\lfloor (d-\lfloor N/2\rfloor)/4\rfloor$, we can obtain $\delta(t_1+\epsilon)\leq(1-l/2)\delta(t_1)$ in \emph{Scenarios} \emph{\uppercase\expandafter{\romannumeral1}.1}, \emph{\uppercase\expandafter{\romannumeral1}.2}, and \emph{\uppercase\expandafter{\romannumeral1}.3}, respectively. At $t_1+\epsilon$, all legitimate oscillators reside in $[0,\pi]$ and will evolve towards phase $2\pi$ rad and fire. By repeating the above analyses, we can get that the length of the containing arc of legitimate oscillators $\delta(t)$ will decrease to a value no greater than $(1-l/2)\delta(t)$ after each round of firing until it converges to $0$.

\emph{Case} \uppercase\expandafter{\romannumeral2}~($\pi/2\leq\epsilon<\pi$): Similar to the reasoning in \emph{Case} \uppercase\expandafter{\romannumeral1}, there exists a time instant $t_0>2T$ at which the ending and starting points of the containing arc of legitimate oscillators reside on phases $0$ and $\pi/2\leq\epsilon<\pi$ rad, respectively. After $t_0$, all legitimate oscillators will evolve freely for exactly $T-\epsilon>T/2$ seconds before the starting point of the containing arc of legitimate oscillators reaches phase $2\pi$ rad at time $t_1=t_0+T-\epsilon$. At $t_1$, the ending point of the containing arc resides on phase $2\pi-\epsilon$ rad and we have $\delta(t_1)=\delta(t_0)=\epsilon$. As the PRF in (\ref{PRF}) is non-negative on $[2\pi-\epsilon,\,2\pi]$, a pulse can only trigger a forward jump or have no effect on a legitimate oscillator with phase in $[2\pi-\epsilon,\,2\pi]$. So all legitimate oscillators will reach phase $2\pi$ rad and fire no later than $t_1+\epsilon$ and within $[t_1,\,t_1+\epsilon/2]$, we can only have one of the following three scenarios:
\leftmargini=22mm
\begin{enumerate}
	\item[\emph{Scenario \uppercase\expandafter{\romannumeral2}.1:}] all legitimate oscillators fired within $[t_1,\,t_1+\epsilon/2]$;
	
	\item[\emph{Scenario \uppercase\expandafter{\romannumeral2}.2:}] some legitimate oscillators did not fire within $[t_1,\,t_1+\epsilon/2]$ but all of these legitimate oscillators jumped in phase within $[t_1,\,t_1+\epsilon/2]$;
	
	\item[\emph{Scenario \uppercase\expandafter{\romannumeral2}.3:}] some legitimate oscillators neither fired nor jumped in phase within $[t_1,\,t_1+\epsilon/2]$.
\end{enumerate}	

Next, we show that $\delta(t)$ will reduce to less than $\pi/2$ rad in finite time, i.e., \emph{Case} \uppercase\expandafter{\romannumeral2} will shift to \emph{Case} \uppercase\expandafter{\romannumeral1} in finite time, after which  $\delta(t)$ will convergence to zero, as ready proven in \emph{Case} \uppercase\expandafter{\romannumeral1}.

Following the same line of reasoning as in \emph{Scenario 2.1} and \emph{Scenario 2.2} of Theorem $1$, one can obtain $\delta(t_1+\epsilon)\leq(1-l/2)\delta(t_1)$ in \emph{Scenario \uppercase\expandafter{\romannumeral2}.1} and \emph{Scenario \uppercase\expandafter{\romannumeral2}.2}, respectively. For \emph{Scenario \uppercase\expandafter{\romannumeral2}.3}, i.e., some legitimate oscillators neither fired nor jumped in phase within $[t_1,\,t_1+\epsilon/2]$, we assume legitimate oscillator $j'$  is such an oscillator. According to Mechanism $1$, there could be two reasons for the not firing of oscillator $j'$ in $[t_1,\,t_1+\epsilon/2]$:
\leftmargini=23.8mm
\leftmarginii=23.8mm
\begin{enumerate}
	\item[\emph{Scenario \uppercase\expandafter{\romannumeral2}.3.1:}] legitimate oscillator $j'$ receives no greater than $\lambda_{j'}$ pulses within $[t_1,\,t_1+\epsilon/2]$, i.e., condition $b)$ of Mechanism $1$ is not satisfied;
	
	\item[\emph{Scenario \uppercase\expandafter{\romannumeral2}.3.2:}] legitimate oscillator $j'$ receives over $\lambda_{j'}$ pulses within $[t_1,\,t_1+\epsilon/2]$, but the number of pulses it received within the past period of length $3T/4$ is no less than $\bar\lambda_{j'}$, i.e., condition $c)$ of Mechanism $1$ is not satisfied.
\end{enumerate}
%%%%%%%%%%%%%%%%%%%%%%%%%%%%%%%%%%%%%%%%%%%%%%%%%%%%%%%%%%%%%%%%%%%%%%%%

Still following the same line of reasoning as in \emph{Scenario 2.3.1} and \emph{Scenario 2.3.2} of Theorem $1$ and using the fact that the number of attackers $M$ is no greater than $2\times\lfloor (d-\lfloor N/2\rfloor)/4\rfloor$, we can obtain in \emph{Scenario \uppercase\expandafter{\romannumeral2}.3.1} that the length of the containing arc of legitimate oscillators will reduce to $(1-l/2)$ of its original value after every firing round whereas in \emph{Scenario \uppercase\expandafter{\romannumeral2}.3.2} such a reduction occurs after every two consecutive firing rounds.

Since every legitimate oscillator will fire at least once within any time interval of length $3T/2$ according to the reasoning at the beginning of Part \uppercase\expandafter{\romannumeral2}, we can get that the length of the containing arc of legitimate oscillators $\delta(t)$ will always decrease to a value less than $\pi/2$ rad within finite time (in fact, after at most $2m$ firing rounds with $m$ satisfying $(1-l/2)^{m}\delta(t_0)<\pi/2$), after which it will converge to zero according to the argument in \emph{Case} \uppercase\expandafter{\romannumeral1}.

By combining \emph{Case} \uppercase\expandafter{\romannumeral1} and \emph{Case} \uppercase\expandafter{\romannumeral2}, one can obtain that the containing arc of legitimate oscillators $\delta(t)$ will always converge to $0$ even in the presence of attackers.
\end{proof}

\begin{Corollary}\label{Corollary_2_2pi_period}
Under conditions in Theorem \ref{Theorem_2_WithAttacks_Noclue}, Mechanism $1$ guarantees that all legitimate oscillators synchronize with an oscillation period $T=2\pi$ seconds even in the presence of attacks. 
\end{Corollary}

\begin{proof}
According to the proof of Theorem \ref{Theorem_2_WithAttacks_Noclue}, we know that despite the presence of attacks, the containing arc of legitimate oscillators will shrink to $0$ upon which the phases of legitimate oscillators will not be affected by attack pulses. Therefore, Mechanism $1$ can guarantee the $T=2\pi$ seconds oscillation period even in the presence of attacks.  
\end{proof}

Next, we prove that Mechanism $1$ can guarantee synchronization of general connected PCO network even when attackers collude with each other and exchange perceived phase information of their neighbors. In this situation, an attacker will emit a malicious pulse either when the pulse can enlarge the containing arc of the union set of colluding attackers' neighbor sets, or when the pulse can help other attack pulse to do so. 

To facilitate the analysis, we first characterize the phase evolution of legitimate oscillators in the presence of colluding attackers. 

\begin{Lemma}\label{Lemma_5_encounter_attack}
For a general connected network of $N$ PCOs, within which $M\leq\lfloor(d-\lfloor N/2\rfloor )/4\rfloor$ oscillators are compromised colluding attackers launching attacks following the stealthy Byzantine attack model in Section \uppercase\expandafter{\romannumeral4}, if the initial length of the containing arc is less than $\pi$ rad and $d>\lfloor N/2 \rfloor$, then under Mechanism $1$, the $N-M$ legitimate oscillators will encounter attack pulses only when their phases reside partially in $[0,\,\pi)$, partially in $(\pi,\,2\pi]$ with phase $2\pi$ rad belonging to the containing arc.
\end{Lemma}

\begin{proof} Similar to Lemma \ref{Lemma_4_encounter_attack}, we know that the phase distribution of legitimate oscillators after the first free-running period can only fall within one of the four scenarios in Fig. \ref{keep_firing}. 
	
According to the stealth Byzantine attack model in Section \uppercase\expandafter{\romannumeral4}, we know that $M$ attackers can emit at most $M$ attack pulses in a quarter period. Given $M\leq\lfloor(d-\lfloor N/2\rfloor)/4\rfloor\leq\lfloor(d(i)-\lfloor N/2\rfloor )/4\rfloor=\lambda_i$ for $i\in\mathcal{N}_L$ where $\mathcal{N}_L$ is the index set of all legitimate oscillators, we know from Mechanism $1$ that attacks pulses alone are not enough to trigger a legitimate oscillator to jump in phase. Therefore, following an argument similar to Lemma \ref{Lemma_4_encounter_attack}, we know that to enlarge the containing arc of legitimate neighbors, attack pulses are sent only when the phases of legitimate oscillators reside partially in $[0,\,\pi)$, partially in $(\pi,\,2\pi]$ with phase $2\pi$ rad belonging to the containing arc.
\end{proof}

Next, we establish the synchronization condition for general connected PCO networks in the presence of colluding attackers.

\begin{Theorem}\label{Theorem_3_WithAttacks_clue_N_known}
For a general-connected network of $N$ PCOs, within which $M\leq\lfloor(d-\lfloor N/2\rfloor)/4\rfloor$ oscillators are colluding attackers launching attacks following the stealthy Byzantine attack model in Sec. \uppercase\expandafter{\romannumeral4}, if the initial length of the containing arc is less than $\pi$ rad and $d>\lfloor N/2 \rfloor$, then all legitimate oscillators can be synchronized under Mechanism $1$.
\end{Theorem}

\begin{proof} Similar to the proof in Theorem \ref{Theorem_2_WithAttacks_Noclue}, we divide the proof into two parts. In Part \uppercase\expandafter{\romannumeral1}, we prove that the length of the containing arc of legitimate oscillators is non-increasing. In Part \uppercase\expandafter{\romannumeral2}, we prove that it will converge to $0$. 

Part \uppercase\expandafter{\romannumeral1} \emph{(The length of the containing arc of legitimate oscillators is non-increasing)}: It can be easily inferred that the length of the containing arc of legitimate oscillators remains unchanged if no legitimate oscillators jump in phase. So we only consider the case where a pulse (from either a legitimate oscillator or an attacker) triggers a phase jump on a legitimate oscillator. 

Following the same line of reasoning as in Theorem \ref{Theorem_2_WithAttacks_Noclue}, one can easily get that the firing of a legitimate oscillator cannot increase the length of the containing arc of legitimate oscillators. By combining Lemma \ref{Lemma_1} and Lemma \ref{Lemma_5_encounter_attack}, we can also obtain that no attacker pulses can increase the length of the containing arc of legitimate oscillators, although the containing arc of a subset of legitimate oscillators (the union set of colluding attackers' neighbor sets) may be enlarged. Hence, we can conclude that the length of the containing arc of all legitimate oscillators is non-increasing.

Part \uppercase\expandafter{\romannumeral2} \emph{(The length of the containing arc of legitimate oscillators converges to $0$)}: The proof follows the same reasoning as in Part \uppercase\expandafter{\romannumeral2} of Theorem \ref{Theorem_2_WithAttacks_Noclue} and is omitted.
\end{proof}

\begin{Remark}
It is worth noting that the maximally allowable number of attackers in a PCO network is $2\times\lfloor (d-\lfloor N/2\rfloor )/4\rfloor$ when attackers do not collude with each other, which is greater than the maximally allowable number of compromised oscillators $\lfloor (d-\lfloor N/2\rfloor )/4\rfloor$ when attackers collude and exchange information.
\end{Remark}

In the colluding case, some attackers can emit attack pulses even if these pulses themselves do not enlarge the containing arc (as long as these pulses can help other attack pulses to enlarge the containing arc). In fact, even if all attackers are allowed to send attack pulses when the containing arc does not change, they still cannot prevent legitimate pulses from satisfying condition (\ref{upper_bound}) to decrease the length of the containing arc.

\begin{Corollary}\label{Corollary_3}
For a general connected network of $N$ PCOs, within which $M\leq\lfloor(d-\lfloor N/2\rfloor )/4\rfloor$ colluding attackers have the ability to emit attack pulses not only when their pulses can enlarge the length of the containing arc but also when the pulses do not change the containing arc, if the initial length of the containing arc of all legitimate oscillators is less than $\pi$ rad and $d>\lfloor N/2 \rfloor$, then there always exist legitimate pulses satisfying (\ref{upper_bound}) in Mechanism $1$.
\end{Corollary}

\begin{proof}According the stealthy requirement in Sec. \uppercase\expandafter{\romannumeral4}, $M\leq\lfloor(d-\lfloor N/2\rfloor )/4\rfloor$ attackers can emit at most $2M$ attack pulses within an arbitrary three-quarter oscillation period. Since $2M$ is less than $\bar{\lambda}_i$, one can get that (\ref{upper_bound}) cannot be made unsatisfied for all legitimate pulses.
\end{proof}

\begin{Remark}
Following Corollary \ref{Corollary_3} and the proof in Theorem $2$, one can get that there always exist legitimate pulses satisfying condition (\ref{upper_bound}), which will reduce the length of the containing arc, even though attackers can ensure that all their attack pulses do not change the length of the containing arc of legitimate oscillators. Hence, attackers cannot prevent legitimate oscillators from reaching synchronization by holding the containing arc constant.
\end{Remark}

\section{Extension to the Case where $N$ is Unknown to Individual Oscillators}

The implementation of the ``cut-off" algorithm in Mechanism $1$ requires each node to have access to $N$, which may be not feasible in a completely decentralized network. Therefore, in this section, we generalize our approach to the case where $N$ is unknown to individual oscillators by leveraging the degree information of individual oscillators. The essence is a new ``cut-off" mechanism that is designed based on the degree information of individual oscillators, as detailed below:

\noindent\rule{9cm}{0.12em}
\emph{New Pulse-Based Synchronization Mechanism (Mechanism 2):}\\
\rule{9cm}{0.1em}
\begin{enumerate}
\item The phase $\phi_i$ of oscillator $i$ evolves from $0$ to $2\pi$ rad with a constant speed $\omega=1$ rad/second.
	
\item Once $\phi_i$ reaches $2\pi$ rad, oscillator $i$ fires and resets its phase to $0$.
	
\item When oscillator $i$ receives a pulse at time instant $t$, it simultaneously resets its phase according to (\ref{PhaseJump}) only when all the following three conditions are satisfied:
\begin{enumerate}
\item an entire period $T=2\pi$ seconds has elapsed since initiation.
		
\item before receiving the current pulse, oscillator $i$ has received at least $\lfloor d(i)/9 \rfloor$ pulses within $(t-T/4,\,t]$, where $\lfloor\bullet\rfloor$ is the largest integer no greater than $``\bullet."$
		
\item before receiving the current pulse, oscillator $i$ has received less than $d(i)-2\times\lfloor d(i)/9 \rfloor$ pulses within $(t-3T/4,\,t]$.
\end{enumerate}
Otherwise, the pulse has no effect on $\phi_i$.
\end{enumerate}
\rule{9cm}{0.12em}

Following a similar line of reasoning in Section \uppercase\expandafter{\romannumeral3} and Section \uppercase\expandafter{\romannumeral5}, we can prove that Mechanism $2$ can synchronize legitimate oscillators both in the absence and presence of attackers.

\begin{Corollary}\label{Corollary_4}
For an attack-free general-connected network of $N$ PCOs, if the degree of the network satisfies $d>\lfloor 2N/3 \rfloor$ and the initial length of the containing arc is less than $\pi$ rad, then all oscillators can be synchronized under Mechanism $2$.
\end{Corollary}
\begin{proof}
Proof of Corollary \ref{Corollary_4} can be obtained following Theorem \ref{Theorem_1_WithoutAttack} and is omitted. 	
\end{proof}

\begin{Theorem}\label{Theorem_4}
For a general connected network of $N$ PCOs, within which $M$ oscillators are non-colluding stealthy Byzantine attackers, if $M$ is no greater than $2\times\lfloor d/9 \rfloor$ with $d>\lfloor2N/3\rfloor$, then all legitimate oscillators can be synchronized under Mechanism $2$ as long as their initial length of the containing arc is less than $\pi$ rad.
\end{Theorem}
\begin{proof}The proof follows the same line of reasoning as in Theorem \ref{Theorem_2_WithAttacks_Noclue}. More specifically, using the same arguments as Part \uppercase\expandafter{\romannumeral1} of Theorem \ref{Theorem_2_WithAttacks_Noclue}, we can obtain that a pulse from neither a legitimate oscillator nor a stealthy Byzantine attacker could enlarge the containing arc of legitimate oscillators under Mechanism $2$, i.e, the length of the containing arc of legitimate oscillators is non-increasing. Then, following the same argument as in Part \uppercase\expandafter{\romannumeral2} of Theorem \ref{Theorem_2_WithAttacks_Noclue}, we know that if $d>\lfloor 2N/3\rfloor$ and $M\leq 2\times\lfloor d/9 \rfloor$ hold, the length of the containing arc of legitimate oscillators will keep decreasing until it converges to $0$.
\end{proof}

\begin{Theorem}\label{Theorem_5}
For a general connected network of $N$ PCOs, within which $M$ oscillators are colluding stealthy Byzantine attackers, if $M$ is no greater than $\lfloor d/9 \rfloor$ with $d>\lfloor2N/3\rfloor$, then all legitimate oscillators can be synchronized under Mechanism $2$ as long as their initial length of the containing arc is less than $\pi$ rad.
\end{Theorem}

\begin{proof}
The proof can be obtained following the same line of argument as in Theorem \ref{Theorem_3_WithAttacks_clue_N_known} and is omitted.	
\end{proof}

\begin{Remark}
When $N$ is unknown to individual oscillators, $d$ has to be over $\lfloor 2N/3\rfloor$, which is greater than $\lfloor N/2\rfloor$ in the case where $N$ is known. The increased requirement on the connectivity of PCO networks is intuitive in that less knowledge of a PCO network requires stronger conditions to guarantee synchronization.
\end{Remark}

Table $1$ summarizes the conditions for Mechanism $1$ and Mechanism $2$ to achieve synchronization.

\begin{table*}[htbp]\label{Table_Global}
	\centering
	\small
	\centering{Table $1$. Synchronization conditions of Mechanism $1$ and Mechanism $2$ ($N$ denotes the total number of oscillators)}
	\begin{center}
		\begin{threeparttable}	
			\renewcommand\arraystretch{2}
			\scalebox{1}{\begin{tabular}{|c |c |c | c|c|c|}\hline
					\makecell{}&
					\makecell{Initial containing \\ arc length}& 
					\makecell{Degree of \\ network $d$}&  
					\makecell{Need knowledge \\  of $N$} &
					\makecell{Number of attackers $M$\\ (non-colluding case)}&
					\makecell{Number of attackers $M$\\ (colluding case)}\\ \hline 
					\makecell{Mechanism 1}& 
					\makecell{less than $\pi$} & 
					\makecell{$d>\lfloor N/2 \rfloor$} & 
					\makecell{Yes} &
				    \makecell{$M\leq2\times\lfloor (d-\lfloor N/2\rfloor)/4\rfloor$}&
					\makecell{$M\leq\lfloor (d-\lfloor N/2\rfloor)/4\rfloor$}\\ \hline 						
					\makecell{Mechanism 2}& 
					\makecell{less than $\pi$} & 
					\makecell{$d>\lfloor 2N/3 \rfloor$} & 
					\makecell{No} &  
					\makecell{$M\leq2\times\lfloor d/9\rfloor\rfloor$} & 
					\makecell{$M\leq\lfloor d/9\rfloor\rfloor$}\\ \hline
				\end{tabular}}
			\end{threeparttable}
		\end{center}	
	\end{table*}
	
%%%%%%%%%%%%%%%%%%%%%%%%%%%%%%%%%%%%%%%%%%%%%%%%%%%%%%%%%%%%%%%%%%%%%%%%%%%%%%%%%

\section{Simulations}
Consider a network of $30$ PCOs distributed on a two-dimension plane as illustrated in Fig. \ref{Topology}. Two oscillators in the network can communicate with each other if and only if their distance is no more than $50$ meters. Thus, the degree of the network is $d=24$. We set the initial time to $t=0$ and chose phases of oscillators randomly from $[0,\,\pi)$. Hence, the initial length of the containing arc satisfied $\delta(0)<\pi$. 

\begin{figure}[htbp]
	\centering
	\includegraphics[width=0.25\textwidth]{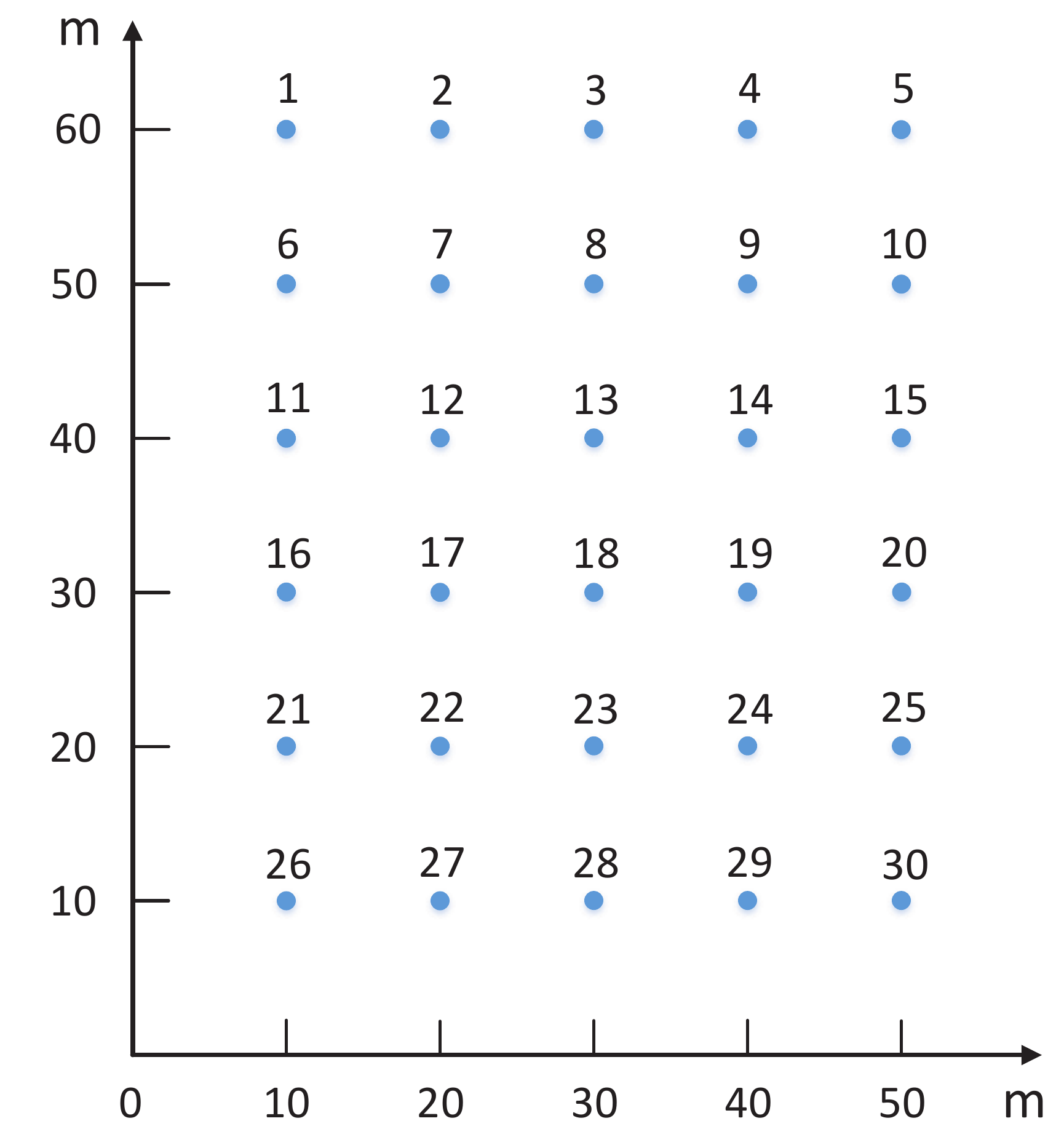}
	\caption{The deployment of the $30$ oscillators used in simulations.}
	\label{Topology}
\end{figure}

\subsection{In the Absence of Attacks}
We first considered the situation without attackers. As $d>\lfloor 2N/3\rfloor=20$, we know from Theorem \ref{Theorem_1_WithoutAttack} and Corollary \ref{Corollary_4} that the network will always synchronize, whether or not $N$ is available to individual oscillators. This was confirmed in Fig. \ref{No_attack_N_known}. 

\begin{figure}[htbp]
	\centering
	\includegraphics[width=0.4\textwidth]{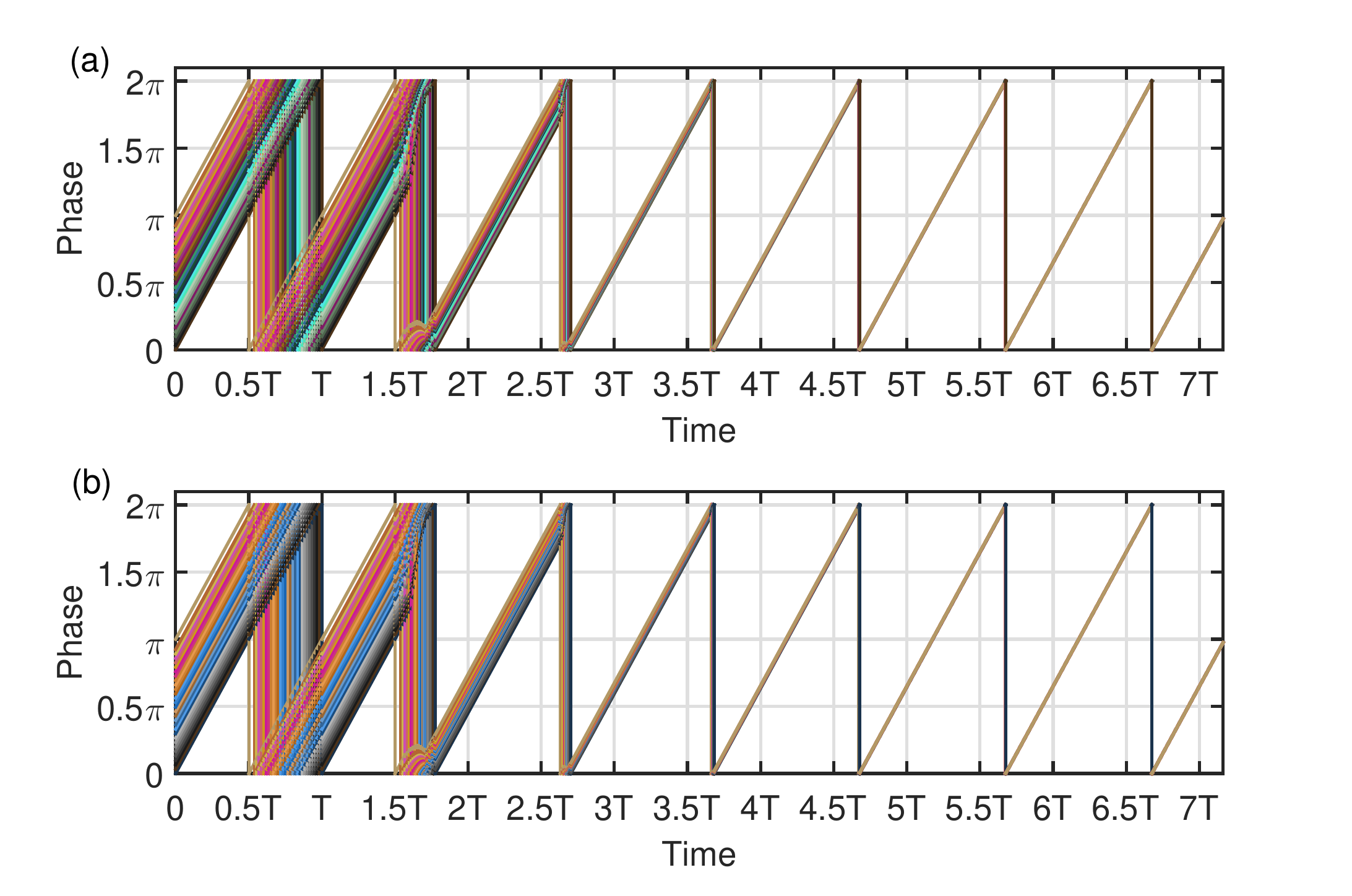}
	\caption{Plot $(a)$ and $(b)$ presented the phase evolutions of the $30$ PCOs under Mechanism $1$ and Mechanism $2$, respectively. The coupling strength was set to $l=0.1$.}
	\label{No_attack_N_known}
\end{figure}

\subsection{In the Presence of Stealthy Byzantine Attackers}

Using the same network, we first ran simulations in the presence of stealthy Byzantine attacks when $N$ is known to individual oscillators.

We assumed that $4$ out of the $30$ oscillators (oscillators $1$, $6$, $26$ and $30$) were compromised and acted as non-colluding Byzantine attackers. As $M=2\times\lfloor(d-\lfloor N/2\rfloor )/4\rfloor=4$, we know from Theorem \ref{Theorem_2_WithAttacks_Noclue} that the network will synchronize. This was confirmed by numerical simulations in Fig. \ref{Attack_N_known_No_clue}, which showed that even under attacks the length of the containing arc of legitimate oscillators converged to zero, despite the fact that the containing arc of oscillator $1$'s legitimate neighbors was enlarged by these attack pulses. 

\begin{figure}[htbp]
	\centering
	\includegraphics[width=0.4\textwidth]{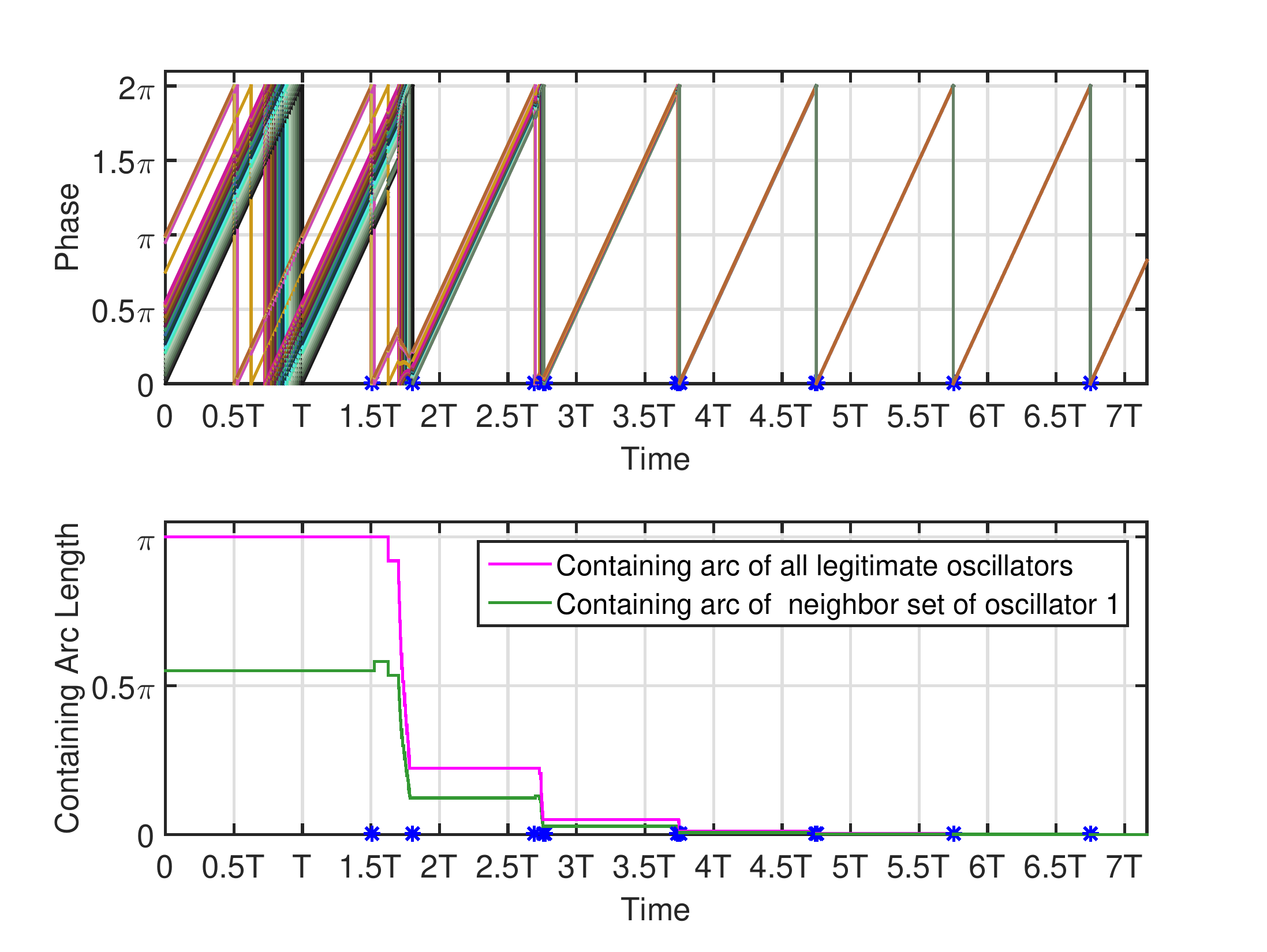}
	\caption{Phase evolution and the length of the containing arc of $26$ legitimate oscillators under Mechanism $1$ in the presence of $4$ non-colluding stealthy Byzantine attackers (oscillators $1$, $6$, $26$, $30$) with attacking pulse time instants represented by asterisks. The coupling strength was set to $l=0.1$.}
	\label{Attack_N_known_No_clue}
\end{figure}

When the $4$ attackers colluded with each other, according to Theorem \ref{Theorem_3_WithAttacks_clue_N_known}, the maximally allowable number of colluding attackers is $\lfloor (d-\lfloor N/2\rfloor )/4\rfloor=2$. Hence, the condition in Theorem \ref{Theorem_3_WithAttacks_clue_N_known} was not satisfied. Simulation results confirmed that legitimate oscillators indeed could not synchronize, as illustrated in Fig. \ref{Attack_N_known_No_clue_Counter}.

\begin{figure}[htbp]
	\centering
	\includegraphics[width=0.4\textwidth]{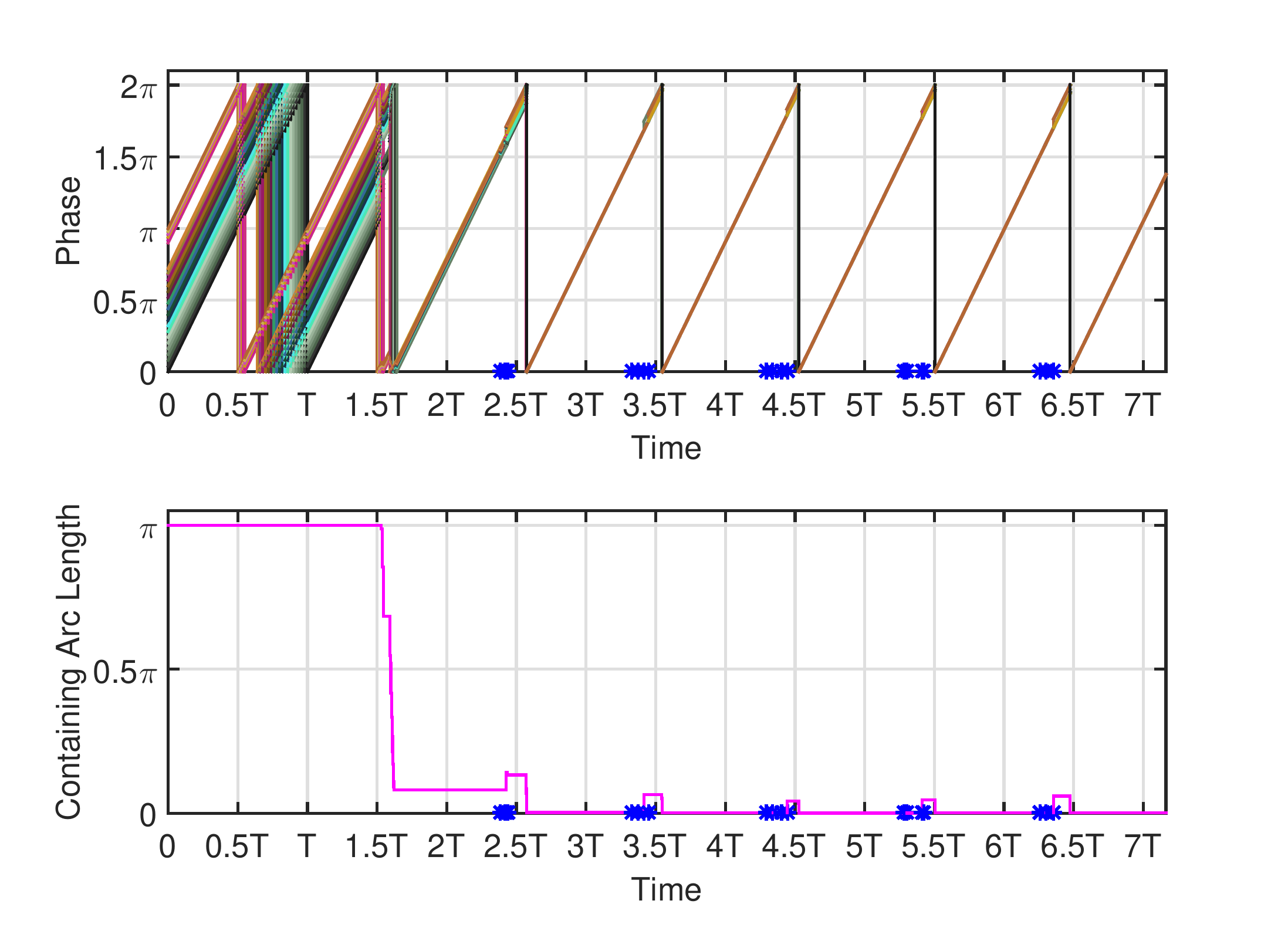}
	\caption{Phase evolution and the length of the containing arc of $26$ legitimate oscillators under Mechanism $1$ in the presence of $4$ colluding stealthy Byzantine attackers (oscillators $1$, $6$, $26$ and $30$) with attacking pulse time instants represented by asterisks. The coupling strength was set to $l=0.1$.}
	\label{Attack_N_known_No_clue_Counter}
\end{figure}

However, when we decreased the number of attackers to $2$ (oscillators $1$ and $6$), all legitimate oscillators synchronized (cf. Fig. \ref{Attack_N_known_clue}), confirming the results in Theorem \ref{Theorem_3_WithAttacks_clue_N_known}. It is worth noting that the containing arc of attacker $1$'s legitimate neighbors were enlarged by attacker pulses, cf. Fig. \ref{Attack_N_known_clue}.

\begin{figure}[htbp]
	\centering
	\includegraphics[width=0.4\textwidth]{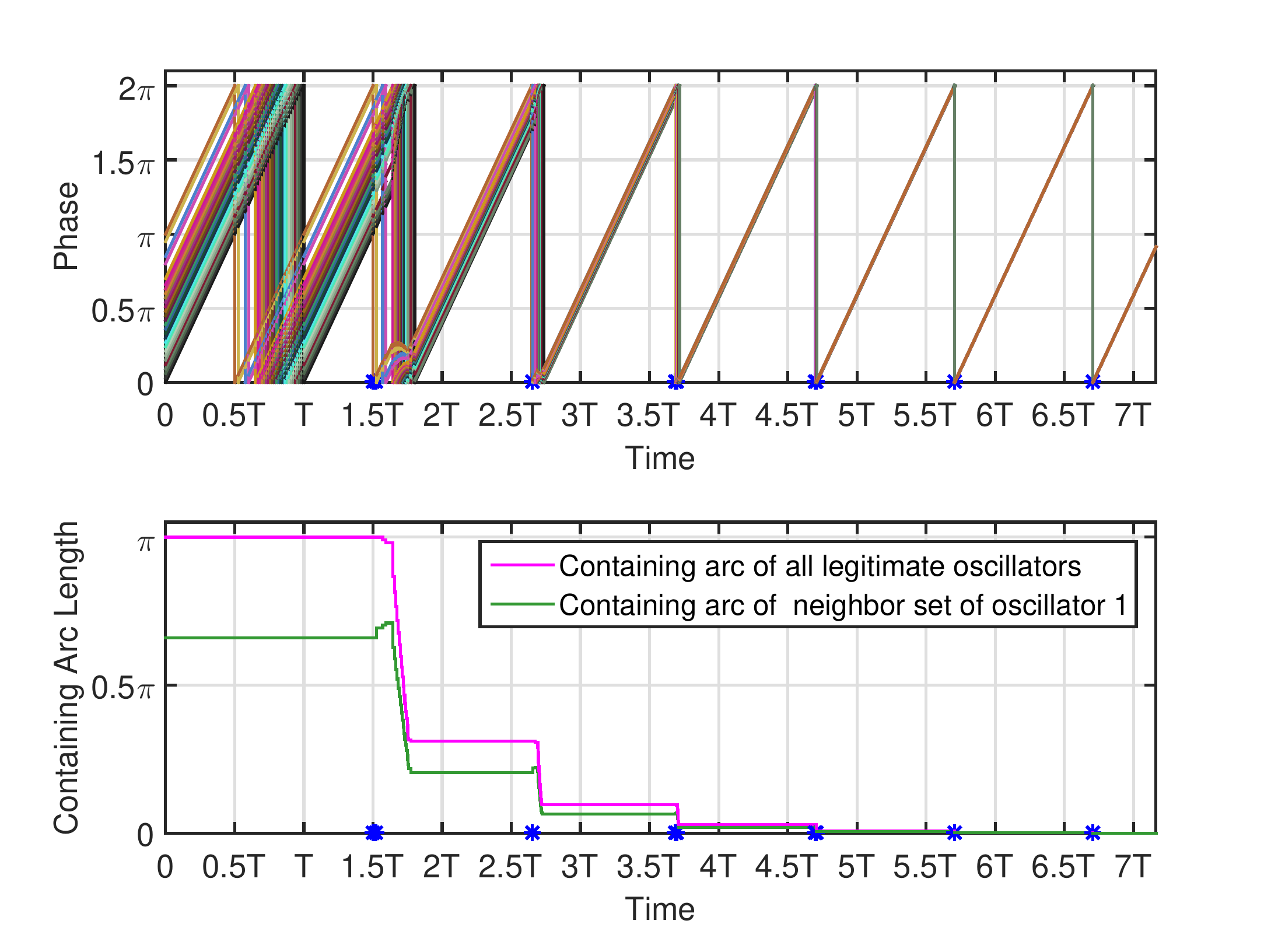}
	\caption{Phase evolution and the length of the containing arc of $28$ legitimate oscillators under Mechanism $1$ in the presence of $2$ colluding stealthy Byzantine attackers (oscillators $1$ and $6$) with attacking pulse time instants represented by asterisks. The coupling strength was set to $l=0.1$.}
	\label{Attack_N_known_clue}
\end{figure}

We also ran simulations in the presence of stealthy Byzantine attacks when $N$ is unknown to individual oscillators. We assumed that $4$ out of the $30$ oscillators (oscillators $1$, $6$, $18$ and $26$) were compromised and acted as stealthy non-colluding Byzantine attackers. According to Theorem \ref{Theorem_4}, all legitimate oscillators can be synchronized under Mechanism $2$. This was confirmed by numerical simulations in Fig. \ref{Attack_N_unknown_No_clue}, which showed that the length of the containing arc of legitimate oscillators converged to zero. 

\begin{figure}[htbp]
	\centering
	\includegraphics[width=0.4\textwidth]{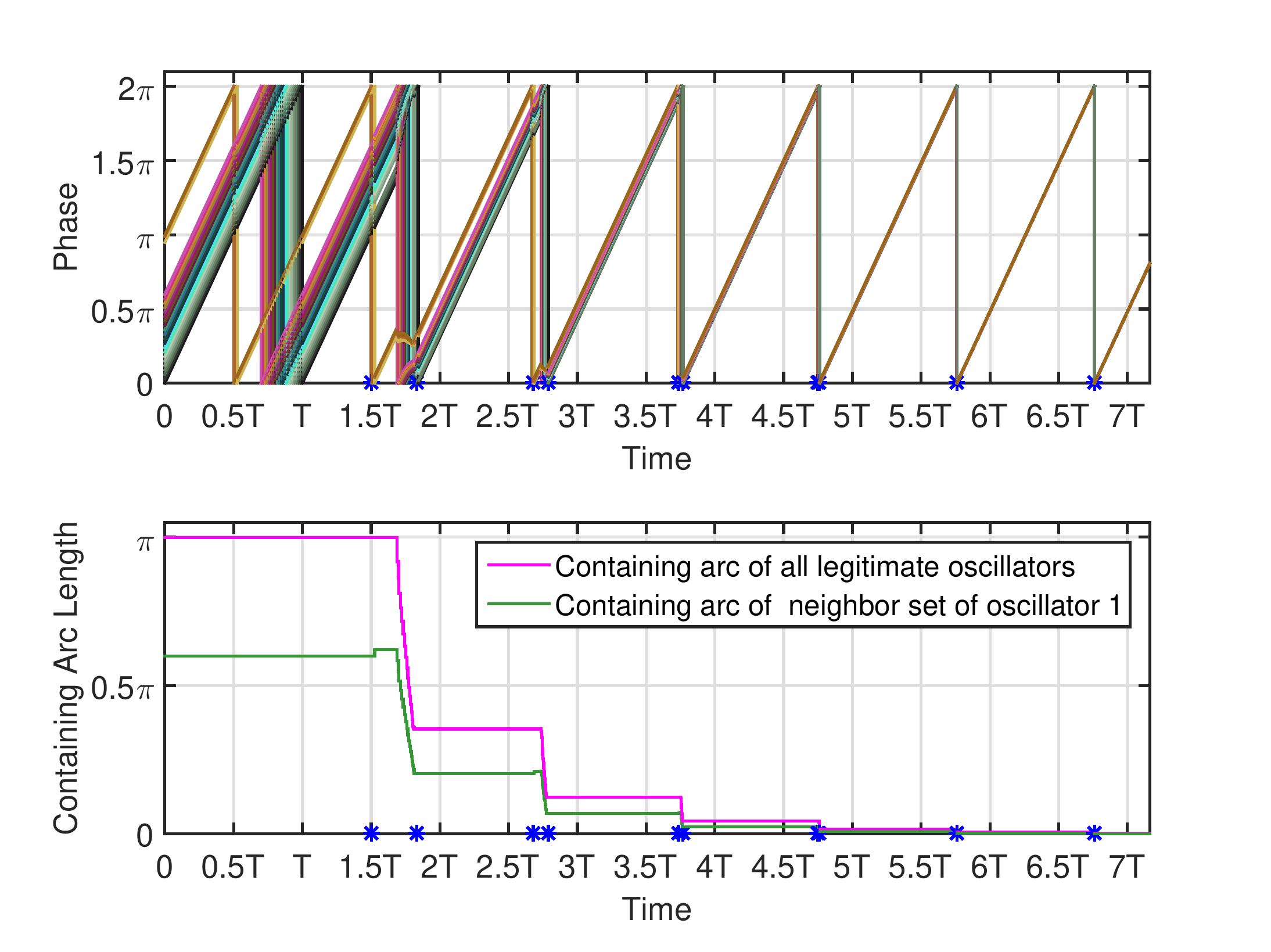}
	\caption{Phase evolution and the length of the containing arc of $26$ legitimate oscillators under Mechanism $2$ in the presence of $4$ stealthy non-colluding Byzantine attackers (oscillators $1$, $6$, $18$ and $26$) with attacking pulse time instants represented by asterisks. The coupling strength was set to $l=0.1$.}
	\label{Attack_N_unknown_No_clue}
\end{figure}

When all $4$ attackers colluded with each other, according to Theorem \ref{Theorem_5}, the maximally allowable number of attackers is $\lfloor d/9\rfloor=2$. Hence, the condition in Theorem \ref{Theorem_5} is not satisfied. Simulation results confirmed that legitimate oscillators indeed could not synchronize, as illustrated in Fig. \ref{Attack_N_unknown_No_clue_Counter}.

\begin{figure}[htbp]
	\centering
	\includegraphics[width=0.4\textwidth]{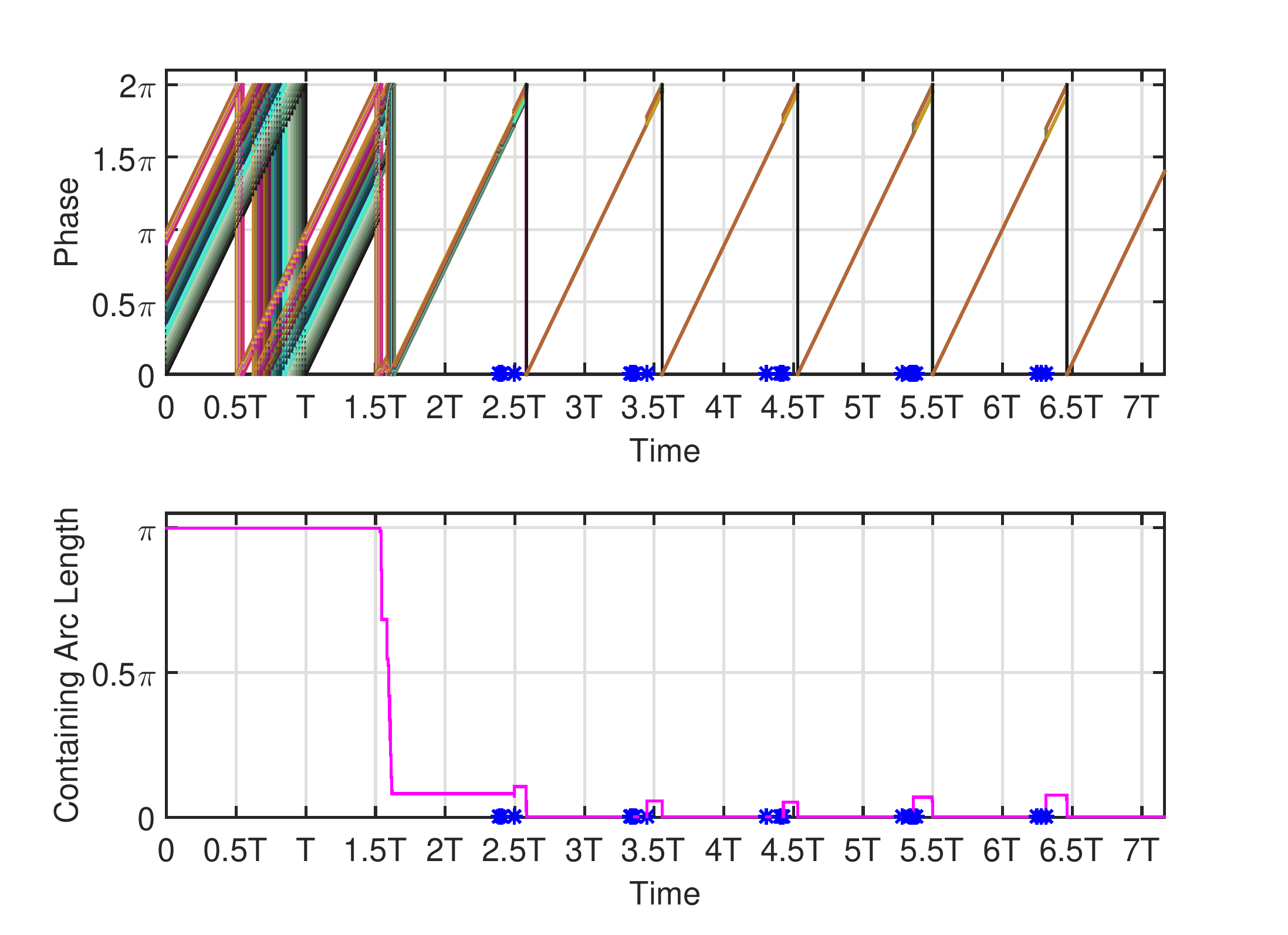}
	\caption{Phase evolution and the length of the containing arc of $26$ legitimate oscillators under Mechanism $2$ in the presence of $4$ colluding stealthy Byzantine attackers (oscillators $1$, $6$, $18$ and $26$) with firing time instants represented by asterisks. The coupling strength was set to $l=0.1$.}
	\label{Attack_N_unknown_No_clue_Counter}
\end{figure}

However, when we reduced the number of colluding attackers to $2$ (oscillators $1$ and $6$), all legitimate oscillators achieved synchronization (cf. Fig. \ref{Attack_N_unknown_clue}), which confirmed Theorem \ref{Theorem_5}.

\begin{figure}[htbp]
	\centering
	\includegraphics[width=0.4\textwidth]{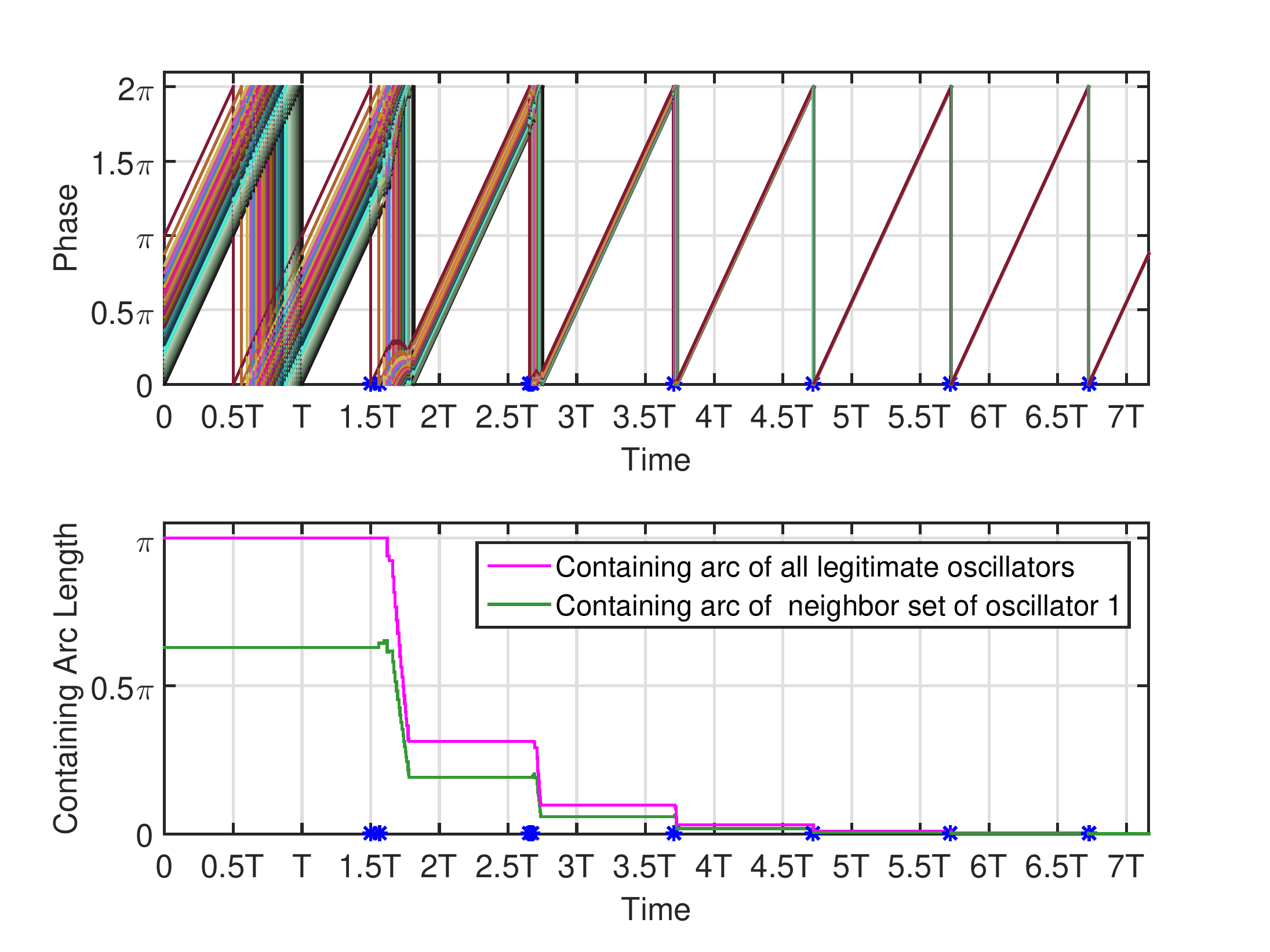}
	\caption{Phase evolution and the length of the containing arc of $28$ legitimate oscillators under Mechanism $2$ in the presence of $2$ colluding stealthy Byzantine attackers (oscillators $1$ and $6$) with attacking pulse time instants represented by asterisks. The coupling strength was set to $l=0.1$.}
	\label{Attack_N_unknown_clue}
\end{figure}

\subsection{Comparison with Existing Results}

In the absence of attacks, we compared Mechanism $1$ with existing approaches in \cite{klinglmayr2012self,yun2015robustness,wang2018pulse} under the PCO network in Fig. \ref{Topology} in the presence of time-varying delays. We assume that the delays are randomly distributed in $[0,\,0.1T]$. Noting that exact synchronization cannot be achieved in this case, similar to \cite{klinglmayr2012self}, we evaluated the performance using synchronization errors defined as follows:
\[
Synchronization~ Error=\max_{i,j\in\mathcal{N}_L}\{\min(2\pi-|\phi_i-\phi_j|,|\phi_i-\phi_j|)\}
\]
where $\mathcal{N}_L$ is the index set of all legitimate oscillators.

Fig. \ref{Sync_error_1} and Fig. \ref{Sync_error_2} show the synchronization errors of Mechanism $1$ and approaches in \cite{klinglmayr2012self,yun2015robustness,wang2018pulse} when the coupling strength was set to $l=0.3$ and $l=0.6$, respectively. Each data point was the average of $10,000$ runs with vertical error bars denoting standard deviations. It can be seen that our approach renders a smaller synchronization error. It is worth noting that Mechanism $2$ also renders a smaller synchronization error than the approaches in \cite{klinglmayr2012self,yun2015robustness,wang2018pulse} under the same set up. However, the results are omitted due to space limitations.

\begin{figure}[htbp]
	\centering
	\includegraphics[width=0.4\textwidth]{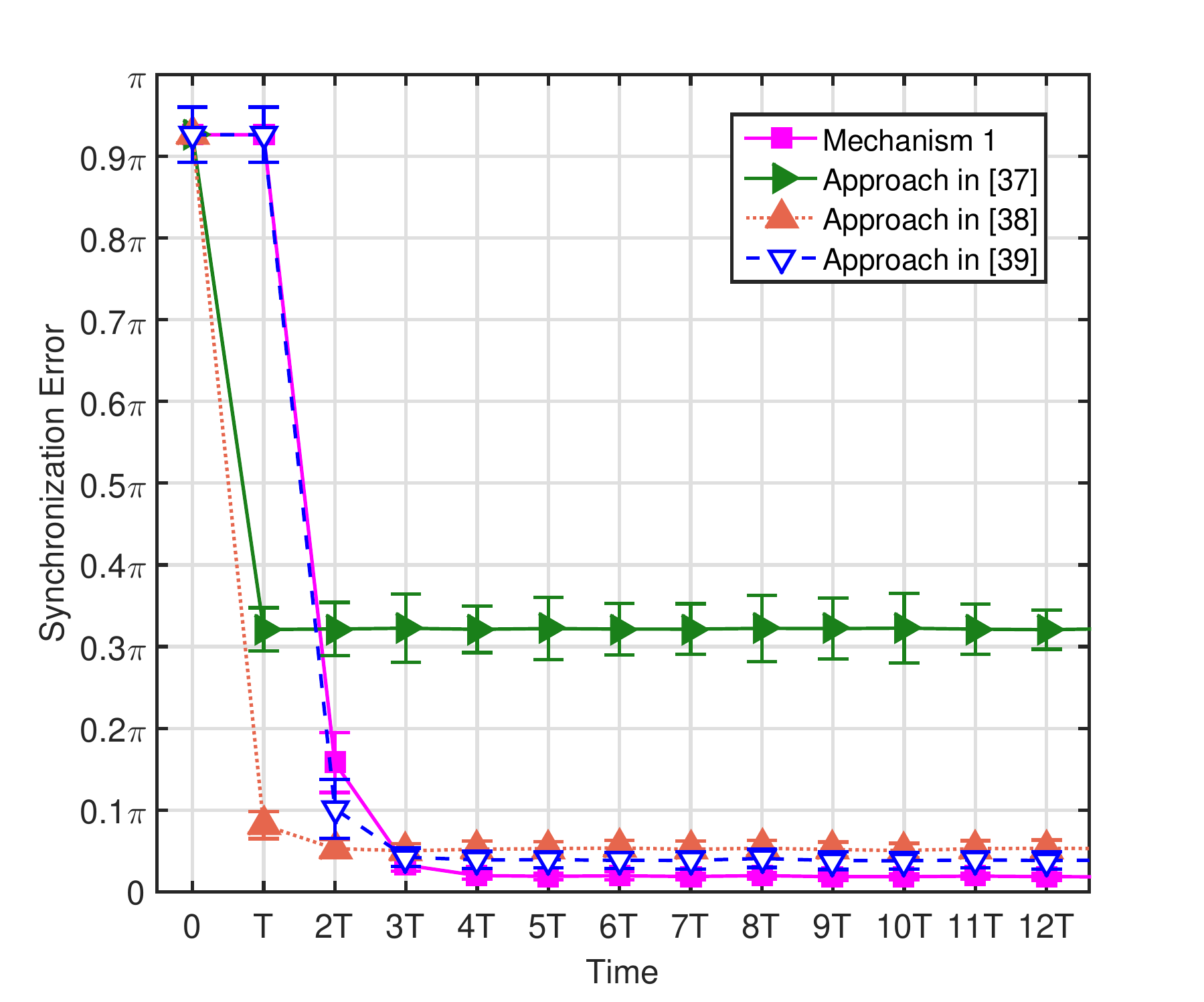}
	\caption{Comparison of our Mechanism $1$ with the approaches in \cite{klinglmayr2012self,yun2015robustness,wang2018pulse} in terms of synchronization error in the presence of time-varying delays uniformly distributed in $[0,\,0.1T]$. The coupling strength was set to $l=0.3$.}
	\label{Sync_error_1}
\end{figure}
\begin{figure}[h]
	\centering
	\includegraphics[width=0.4\textwidth]{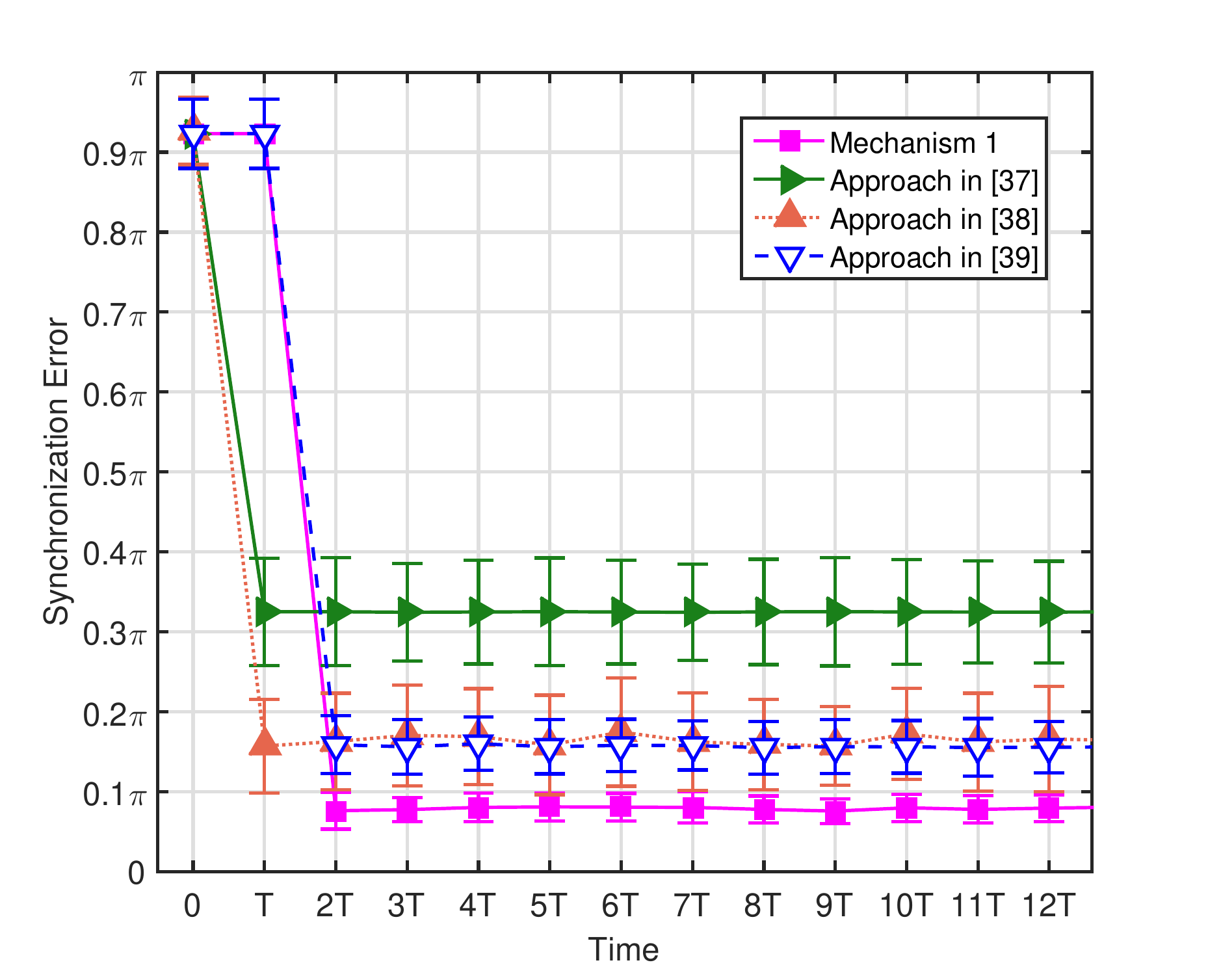}
		 \caption{Comparison of our Mechanism $1$ with the approaches in \cite{klinglmayr2012self,yun2015robustness,wang2018pulse} in terms of synchronization error in the presence of time-varying delays uniformly distributed in $[0,\,0.1T]$. The coupling strength was set to $l=0.6$.}
	\label{Sync_error_2}
\end{figure}

We also compared our proposed approach with existing approaches in \cite{klinglmayr2012self,yun2015robustness,wang2018pulse} under the PCO network in Fig. \ref{Topology} in the presence of non-colluding and colluding stealthy Byzantine attackers, respectively.

Fig. \ref{Sync_error_Non_colluding_N_Known} shows the synchronization errors of Mechanism $1$ and approaches in \cite{klinglmayr2012self,yun2015robustness,wang2018pulse} in the presence of $4$ non-colluding stealthy Byzantine attackers (oscillators $1$, $6$, $26$, $30$) and Fig. \ref{Sync_error_Colluding_N_Known} shows the corresponding synchronization errors in the presence of $2$ colluding stealthy Byzantine attackers (oscillators $1$ and $6$). Each data point was the average of $10,000$ runs with vertical error bars denoting standard deviations. It can be seen that our approach can achieve perfect synchronization whereas all existing approaches are subject to substantial synchronization errors. It is worth noting that our Mechanism $2$ also achieved perfect synchronization under the same set up. However, the results are omitted due to space limitations. 

\begin{figure}[htbp]
	\centering
	\includegraphics[width=0.4\textwidth]{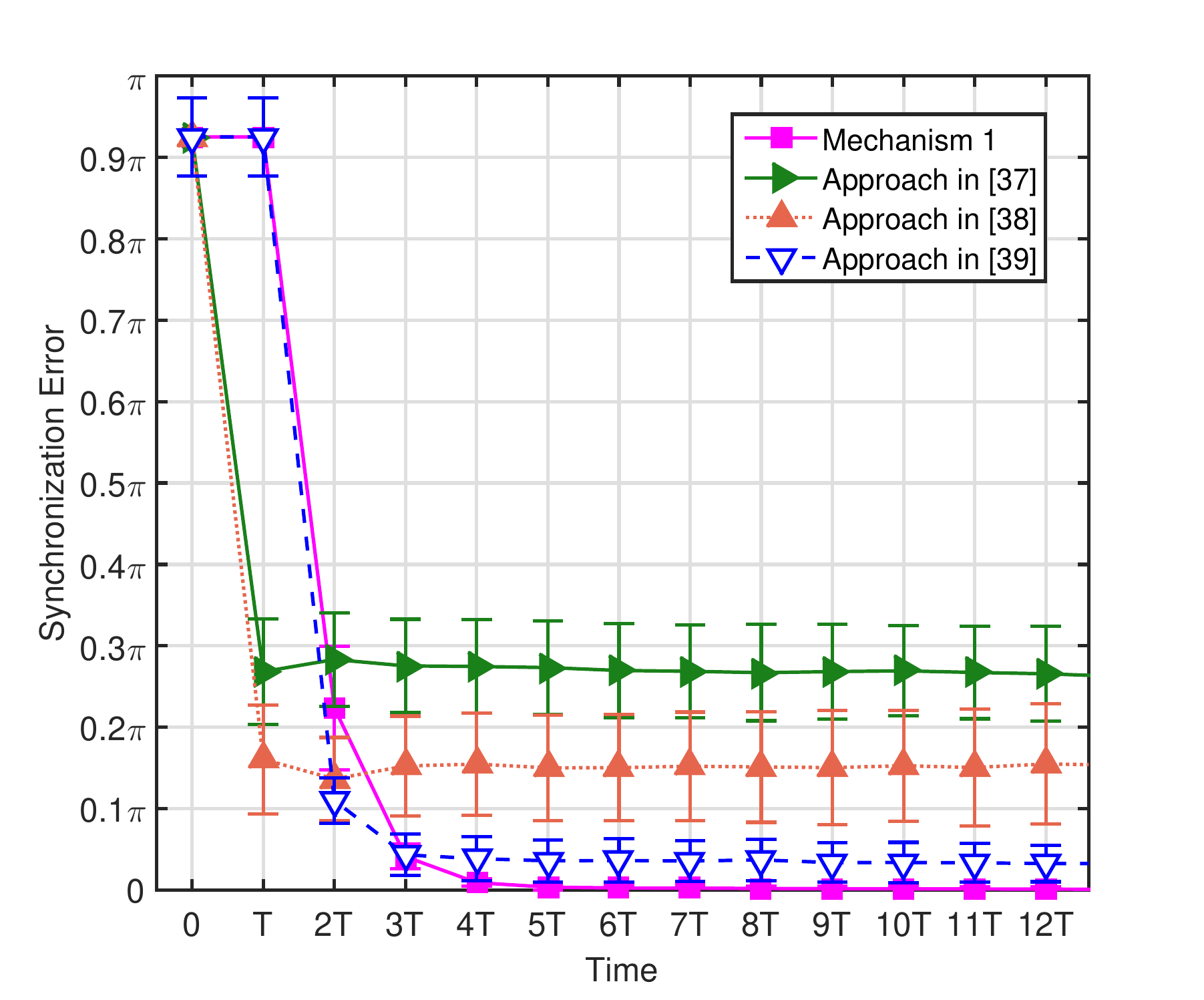}
	\caption{Comparison of our Mechanism $1$ with the attack resilient approaches in \cite{klinglmayr2012self,yun2015robustness,wang2018pulse} in terms of synchronization error in the presence of $4$ non-colluding stealthy Byzantine attackers (oscillators $1$, $6$, $26$, $30$). The coupling strength was set to $l=0.3$.}
	\label{Sync_error_Non_colluding_N_Known}
\end{figure}

\begin{figure}[htbp]
	\centering
	\includegraphics[width=0.4\textwidth]{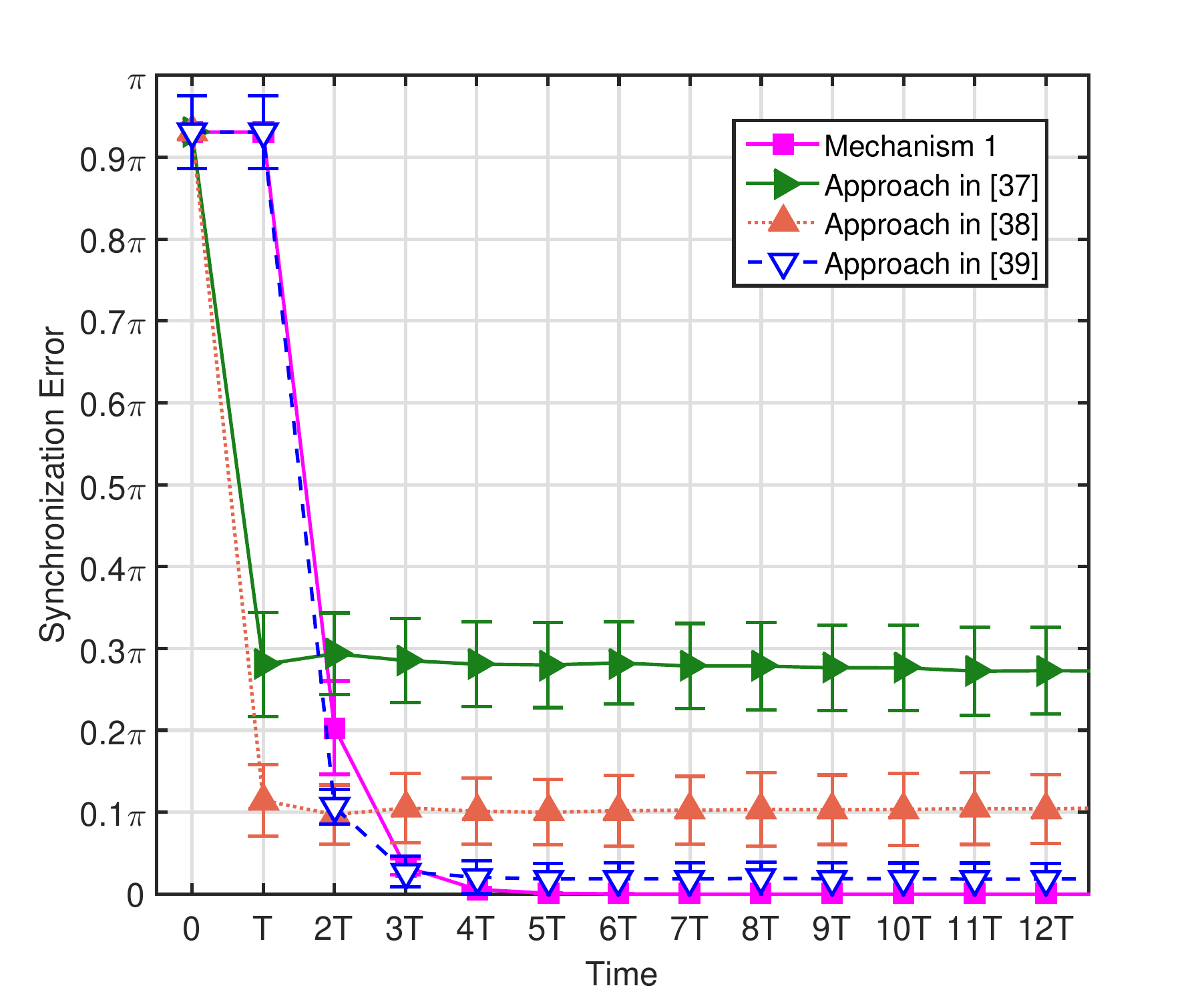}
	\caption{Comparison of our Mechanism $1$ with the attack resilient approaches in \cite{klinglmayr2012self,yun2015robustness,wang2018pulse} in terms of synchronization error in the presence of $2$ colluding stealthy Byzantine attackers (oscillators $1$ and $6$). The coupling strength was set to $l=0.3$.}
	\label{Sync_error_Colluding_N_Known}
\end{figure}

%%%%%%%%%%%%%%%%%%%%%%%%%%%%%%%%%%%%%%%%%%%%%%%%%%%%%%%%%%%%%%%%%%%%%%%%%%%%%%%%%%%%%%%%
\section{CONCLUSIONS}
Due to unique advantages over conventional packet-based synchronization approaches in terms of simplicity, scalability, and energy efficiency, pulse-based synchronization gained increased attention in recent years. However, few results are available to address the attack-resilience of pulse base synchronization. In this paper, we propose a new pulse-based synchronization mechanism to improve the attack-resilience of general connected PCO networks. We rigorously prove that the new mechanism can achieve phase synchronization of general connected PCO networks in the presence of multiple stealthy Byzantine attackers, irrespective of whether they are colluding or not. Our results allow the initial phases of legitimate oscillators to reside in a half oscillation period, which is in distinct difference from most existing attack-resilience algorithms that require a priori (almost) synchronization among legitimate oscillators. The approach is also applicable when the total number of oscillators is unknown to individual oscillators. Numerical simulations confirmed the analytical results.   
%%%%%%%%%%%%%%%%%%%%%%%%%%%%%%%%%%%%%%%%%%%%%%%%%%%%%%%%%%%%%%%%%%%%%%%%%%%%%%%%%%%
\bibliographystyle{unsrt}
\bibliography{Attack}

\begin{thebibliography}{10}

\bibitem{lamport1985synchronizing}
L.~Lamport and P.~M. Melliar-Smith.
\newblock Synchronizing clocks in the presence of faults.
\newblock {\em Journal of the ACM (JACM)}, 32(1):52--78, 1985.

\bibitem{mirollo1990synchronization}
R.~Mirollo and S.~Strogatz.
\newblock Synchronization of pulse-coupled biological oscillators.
\newblock {\em SIAM Journal on Applied Mathematics}, 50(6):1645--1662, 1990.

\bibitem{peskin1975mathematical}
C.~S. Peskin.
\newblock {\em Mathematical aspects of heart physiology}.
\newblock Courant Institute of Mathematical Sciences, New York University,
  1975.

\bibitem{mathar1996pulse}
R.~Mathar and J.~Mattfeldt.
\newblock Pulse-coupled decentral synchronization.
\newblock {\em SIAM Journal on Applied Mathematics}, 56(4):1094--1106, 1996.

\bibitem{simeone2008distributed}
O.~Simeone, U.~Spagnolini, Y.~Bar-Ness, and S.~Strogatz.
\newblock Distributed synchronization in wireless networks.
\newblock {\em IEEE Signal Processing Magazine}, 25(5):81--97, 2008.

\bibitem{pagliari2011scalable}
R.~Pagliari and A.~Scaglione.
\newblock Scalable network synchronization with pulse-coupled oscillators.
\newblock {\em IEEE Transactions on Mobile Computing}, 10(3):392--405, 2011.

\bibitem{werner2005firefly}
G.~Werner-Allen, G.~Tewari, A.~Patel, M.~Welsh, and R.~Nagpal.
\newblock Firefly-inspired sensor network synchronicity with realistic radio
  effects.
\newblock In {\em Proceedings of the 3rd international conference on Embedded
  networked sensor systems}, pages 142--153. ACM, 2005.

\bibitem{hong2005scalable}
Y.~W. Hong and A.~Scaglione.
\newblock A scalable synchronization protocol for large scale sensor networks
  and its applications.
\newblock {\em IEEE Journal on Selected Areas in Communications},
  23(5):1085--1099, 2005.

\bibitem{hu2006scalability}
A.~Hu and S.~D. Servetto.
\newblock On the scalability of cooperative time synchronization in
  pulse-connected networks.
\newblock {\em IEEE Transactions on Information Theory}, 52(6):2725--2748,
  2006.

\bibitem{leidenfrost2009firefly}
R.~Leidenfrost and W.~Elmenreich.
\newblock Firefly clock synchronization in an 802.15. 4 wireless network.
\newblock {\em EURASIP Journal on Embedded Systems}, 2009(1):1, 2009.

\bibitem{nunez2017pulse}
F.~N\'u$\rm\tilde{n}$ez, Y.~Q. Wang, D.~Grasing, S.~Desai, G.~Cakiades, and
  F.~J.~Doyle III.
\newblock Pulse-coupled time synchronization for distributed acoustic event
  detection using wireless sensor networks.
\newblock {\em Control Engineering Practice}, 60:106--117, 2017.

\bibitem{wang_tsp2:12}
Y.~Q. Wang and F.~J. {Doyle III}.
\newblock Optimal phase response functions for fast pulse-coupled
  synchronization in wireless sensor networks.
\newblock {\em IEEE Transactions on Signal Processing}, 60(10):5583--5588,
  2012.

\bibitem{konishi:08}
K.~Konishi and H.~Kokame.
\newblock Synchronization of pulse-coupled oscillators with a refractory period
  and frequency distribution for a wireless sensor network.
\newblock {\em Chaos: An Interdisciplinary Journal of Nonlinear Science},
  18(3):033132, 2008.

\bibitem{okuda2011experimental}
T.~Okuda, K.~Konishi, and N.~Hara.
\newblock Experimental verification of synchronization in pulse-coupled
  oscillators with a refractory period and frequency distribution.
\newblock {\em Chaos: An Interdisciplinary Journal of Nonlinear Science},
  21(2):023105, 2011.

\bibitem{wang_tsp:12}
Y.~Q. Wang, F.~N\'u$\rm\tilde{n}$ez, and F.~J. {Doyle III}.
\newblock Energy-efficient pulse-coupled synchronization strategy design for
  wireless sensor networks through reduced idle listening.
\newblock {\em IEEE Transactions on Signal Processing}, 60(10):5293--5306,
  2012.

\bibitem{wang_tsp:13}
Y.~Q. Wang, F.~N\'u$\rm\tilde{n}$ez, and F.~J.~Doyle III.
\newblock Statistical analysis of the pulse-coupled synchronization strategy
  for wireless sensor networks.
\newblock {\em IEEE Transactions on Signal Processing}, 61(21):5193--5204,
  2013.

\bibitem{nunez2015synchronization}
F.~N\'u$\rm\tilde{n}$ez, Y.~Q. Wang, and F.~J.~Doyle III.
\newblock Synchronization of pulse-coupled oscillators on (strongly) connected
  graphs.
\newblock {\em IEEE Transactions on Automatic Control}, 60(6):1710--1715, 2015.

\bibitem{nunez2016synchronization}
F.~N\'u$\rm\tilde{n}$ez, Y.~Q. Wang, A.~R. Teel, and F.~J.~Doyle III.
\newblock Synchronization of pulse-coupled oscillators to a global pacemaker.
\newblock {\em Systems \& Control Letters}, 88:75--80, 2016.

\bibitem{klinglmayr2012guaranteeing}
J.~Klinglmayr, C.~Kirst, C.~Bettstetter, and M.~Timme.
\newblock Guaranteeing global synchronization in networks with stochastic
  interactions.
\newblock {\em New Journal of Physics}, 14(7):073031, 2012.

\bibitem{klinglmayr2017convergence}
J.~Klinglmayr, C.~Bettstetter, M.~Timme, and C.~Kirst.
\newblock Convergence of self-organizing pulse-coupled oscillator
  synchronization in dynamic networks.
\newblock {\em IEEE Transactions on Automatic Control}, 62(4):1606--1619, 2017.

\bibitem{canavier2010pulse}
C.~Canavier and S.~Achuthan.
\newblock Pulse coupled oscillators and the phase resetting curve.
\newblock {\em Mathematical biosciences}, 226(2):77--96, 2010.

\bibitem{nishimura2011robust}
J.~Nishimura and E.~J. Friedman.
\newblock Robust convergence in pulse-coupled oscillators with delays.
\newblock {\em Physical Review Letters}, 106(19):194101, 2011.

\bibitem{nishimura2012probabilistic}
J.~Nishimura and E.~J. Friedman.
\newblock Probabilistic convergence guarantees for type-ii pulse-coupled
  oscillators.
\newblock {\em Physical Review E}, 86(2):025201, 2012.

\bibitem{lucken2012two}
L.~L{\"u}cken and S.~Yanchuk.
\newblock Two-cluster bifurcations in systems of globally pulse-coupled
  oscillators.
\newblock {\em Physica D: Nonlinear Phenomena}, 241(4):350--359, 2012.

\bibitem{nunez2015global}
F.~N{\'u}{\~n}ez, Y.~Q. Wang, and F.~J. Doyle.
\newblock Global synchronization of pulse-coupled oscillators interacting on
  cycle graphs.
\newblock {\em Automatica}, 52:202--209, 2015.

\bibitem{kannapan2016synchronization}
D.~Kannapan and F.~Bullo.
\newblock Synchronization in pulse-coupled oscillators with delayed
  excitatory/inhibitory coupling.
\newblock {\em SIAM Journal on Control and Optimization}, 54(4):1872--1894,
  2016.

\bibitem{lyu2018global}
H.~Lyu.
\newblock Global synchronization of pulse-coupled oscillators on trees.
\newblock {\em SIAM Journal on Applied Dynamical Systems}, 17(2):1521--1559,
  2018.

\bibitem{proskurnikov2017synchronization}
A.~V. Proskurnikov and M.~Cao.
\newblock Synchronization of pulse-coupled oscillators and clocks under minimal
  connectivity assumptions.
\newblock {\em IEEE Transactions on Automatic Control}, 62(11):5873--5879,
  2016.

\bibitem{gao2018pulse}
H.~Gao and Y.~Q. Wang.
\newblock A pulse-based integrated communication and control design for
  decentralized collective motion coordination.
\newblock {\em IEEE Transactions on Automatic Control}, 63(6):1858--1864, 2018.

\bibitem{pease1980reaching}
M.~Pease, R.~Shostak, and L.~Lamport.
\newblock Reaching agreement in the presence of faults.
\newblock {\em Journal of the ACM (JACM)}, 27(2):228--234, 1980.

\bibitem{manzo2005time}
M.~Manzo, T.~Roosta, and S.~Sastry.
\newblock Time synchronization attacks in sensor networks.
\newblock In {\em Proceedings of the 3rd ACM workshop on Security of ad hoc and
  sensor networks}, pages 107--116. ACM, 2005.

\bibitem{li2006global}
Q.~Li and D.~Rus.
\newblock Global clock synchronization in sensor networks.
\newblock {\em IEEE Transactions on computers}, 55(2):214--226, 2006.

\bibitem{song2007attack}
H.~Song, S.~Zhu, and G.~H. Cao.
\newblock Attack-resilient time synchronization for wireless sensor networks.
\newblock {\em Ad Hoc Networks}, 5(1):112--125, 2007.

\bibitem{du2008security}
X.~J. Du and H.~Chen.
\newblock Security in wireless sensor networks.
\newblock {\em IEEE Wireless Communications}, 15(4), 2008.

\bibitem{leidenfrost2010fault}
R.~Leidenfrost, W.~Elmenreich, and C.~Bettstetter.
\newblock Fault-tolerant averaging for self-organizing synchronization in
  wireless ad hoc networks.
\newblock In {\em 2010 7th International Symposium on Wireless Communication
  Systems}, pages 721--725, 2010.

\bibitem{tyrrell2010does}
A.~Tyrrell, G.~Auer, C.~Bettstetter, and R.~Naripella.
\newblock How does a faulty node disturb decentralized slot synchronization
  over wireless networks?
\newblock In {\em 2010 IEEE International Conference on Communications}, pages
  1--5, 2010.

\bibitem{klinglmayr2012self}
J.~Klinglmayr and C.~Bettstetter.
\newblock Self-organizing synchronization with inhibitory-coupled oscillators:
  Convergence and robustness.
\newblock {\em ACM Transactions on Autonomous and Adaptive Systems (TAAS)},
  7(3):30, 2012.

\bibitem{yun2015robustness}
S.~Yun, J.~Ha, and B.~J. Kwak.
\newblock Robustness of biologically inspired pulse-coupled synchronization
  against static attacks.
\newblock In {\em 2015 IEEE Global Communications Conference (GLOBECOM)}, pages
  1--6. IEEE, 2015.

\bibitem{wang2018pulse}
Z.~Q. Wang and Y.~Q. Wang.
\newblock Pulse-coupled oscillators resilient to stealthy attacks.
\newblock {\em IEEE Transactions on Signal Processing}, 66(12):3086--3099,
  2018.

\bibitem{wang2018attack}
Z.~Q. Wang and Y.~Q. Wang.
\newblock Attack-resilient pulse-coupled synchronization.
\newblock {\em IEEE Transactions on Control of Network Systems}, 6(1):338--351,
  2018.

\bibitem{lamport1982byzantine}
L.~Lamport, R.~Shostak, and M.~Pease.
\newblock The byzantine generals problem.
\newblock {\em ACM Transactions on Programming Languages and Systems (TOPLAS)},
  4(3):382--401, 1982.

\bibitem{xu2005feasibility}
W.~Y. Xu, W.~Trappe, Y.~Y. Zhang, and T.~Wood.
\newblock The feasibility of launching and detecting jamming attacks in
  wireless networks.
\newblock In {\em Proceedings of the 6th ACM international symposium on Mobile
  ad hoc networking and computing}, pages 46--57. ACM, 2005.

\end{thebibliography}
%%%%%%%%%%%%%%%%%%%%%%%%%%%%%%%%%%%%%%%%%%%%%%%%%%%%%%%%%%%%%%%%%%%%%%%%%%%%%%%%%%%

\begin{IEEEbiography}[{\includegraphics[width=1in,height=1.25in,clip,keepaspectratio]{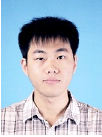}}]{Zhenqian Wang}
	was born in Hebei, China. He received the B.E. and M.Sc. degrees in automation and control theory from Tianjin University, Tianjin, China, in 2012 and 2015, respectively. He is currently working toward the Ph.D. degree in the Department of Electrical and Computer Engineering, Clemson University, Clemson, SC, USA. His current research focuses on attack-resilient clock synchronization.
\end{IEEEbiography}
\begin{IEEEbiography}[{\includegraphics[width=1in,height=1.25in,clip,keepaspectratio]{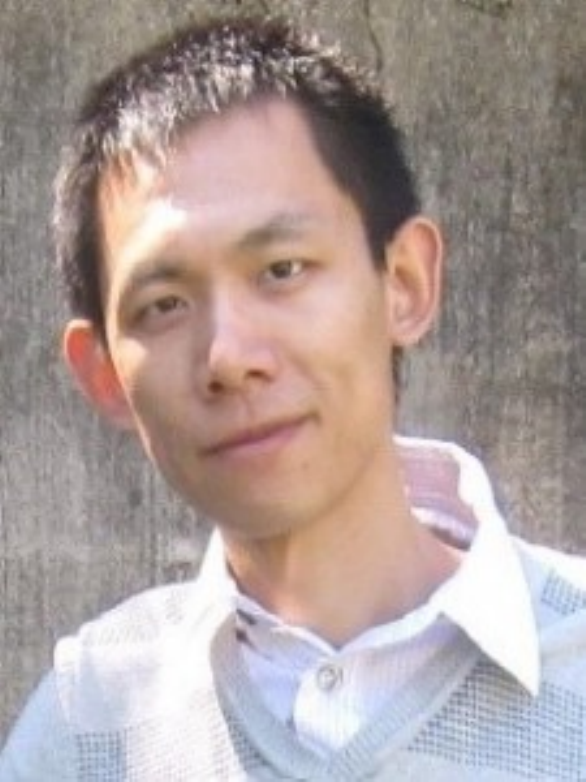}}]{Yongqiang Wang}
	(SM'13) was born in Shandong, China. He received the B.S. degree in Electrical Engineering \& Automation, the B.S. degree in Computer Science \& Technology from Xi'an Jiaotong University, Shaanxi, China, in 2004. He received the M.Sc. and the Ph.D. degrees in Control Science \& Engineering from Tsinghua University, Beijing, China, in 2009. 
	
	From 2007-2008, he was with the University of Duisburg-Essen, Germany, as a visiting student. He was a Project Scientist at the University of California, Santa Barbara. He is currently an Assistant Professor with the Department of Electrical and Computer Engineering, Clemson University, Clemson, SC, USA. His research interests are cooperative and networked control, synchronization of wireless sensor networks, systems modeling and analysis of biochemical oscillator networks, and model-based fault diagnosis. He received the 2008 Young Author Prize from IFAC Japan Foundation for a paper presented at the 17th IFAC World Congress in Seoul.
\end{IEEEbiography}
\end{document}